%% file: draft5.tex
\documentclass[letterpaper,11pt]{article}
\usepackage{amsmath, amsthm, amssymb}
\usepackage[usenames, dvipsnames]{color}
\usepackage[normalem]{ulem} 
\usepackage{fullpage}
\usepackage{cite}
\usepackage[numbers, sort&compress]{natbib}
\usepackage[noend]{algorithmic}
\usepackage{tikz, pgfplots}
\usepackage{enumerate}
\usepackage{hyperref}
\usepackage{xspace,color}
\usepackage{graphicx}
\usepackage{caption}
\usepackage{subcaption}
\usepackage{enumitem,linegoal}
\usepackage[linesnumbered,ruled,vlined]{algorithm2e}
\usepackage{comment}	
\usepackage{ctable}

\usepackage{mathtools}
\usepackage[flushleft]{threeparttable}
\usepackage{verbatim}
\usepackage{setspace}
\usepackage{bbm}
\usepackage{thm-restate}
\usepackage{footnote}
\usepackage{cleveref}
\crefname{LP}{LP}{LPs}
\usepackage{tikz}
\usetikzlibrary{positioning, decorations.pathmorphing}
\usepackage{pgfplots}
\pgfplotsset{compat=1.18}
\usepackage{hyperref}
\usepackage{tikz}
\usetikzlibrary{decorations.pathreplacing}

\makesavenoteenv{tabular}
\makesavenoteenv{table}

\usepackage[margin=1in]{geometry}

\definecolor{crimsonglory}{rgb}{0,0,0}

 \newtheorem{theorem}{Theorem}[section]
 
 \newtheorem{lemma}[theorem]{Lemma}
 
 \newtheorem{corollary}[theorem]{Corollary}

\makeatletter
\def\GrabProofArgument[#1]{ #1: \egroup\ignorespaces}
\def\proof{\noindent\textbf\bgroup Proof%
	\@ifnextchar[{\GrabProofArgument}{. \egroup\ignorespaces}}

\makeatother

\usetikzlibrary{arrows,shapes,snakes,automata,backgrounds,petri,calc}
\usepackage[latin1]{inputenc}

\input{src/macros.tex}

\newcounter{proccnt}

\newcommand{\konote}[1]{}

\usepackage[normalem]{ulem} 

\title{Maximizing Social Influence in Almost Linear Time}

\author{
 Saeed Seddighin
}

\SetCommentSty{mycommfont}

\SetKwInput{KwInput}{Input}                
\SetKwInput{KwOutput}{Output}  
\begin{document}
	\newcommand{\ignore}[1]{}
\renewcommand{\theenumi}{(\roman{enumi}).}
\renewcommand{\labelenumi}{\theenumi}
\sloppy
\date{}

\maketitle

\thispagestyle{empty}

\begin{abstract}
\input{src/abstract.tex}
\end{abstract}
\input{src/parallel.tex}
\input{src/intro.tex}
\input{src/preliminaries.tex}
\input{src/results.tex}
\input{src/nk.tex}
\input{src/curse.tex}
\input{src/strong.tex}
\input{src/query.tex}
\bibliographystyle{unsrtnat}	
\bibliography{draft}
\appendix

\end{document}

%% file: src/macros.tex
\newcounter{claudecnt}

\newif\ifshowclaudeedits
\showclaudeeditsfalse 
\newcommand{\claudeedit}[1]{\ifshowclaudeedits{\color{red}#1}\else{#1}\fi}

%% file: src/abstract.tex
Influence maximization is a central algorithmic challenge in network analysis, aiming to identify a set of $k$ seed nodes in a graph with $n$ nodes and $m$ edges that maximizes the expected cascade of information under standard diffusion models. The seminal work of Borgs, Brautbar, Chayes, and Lucier (SODA'14) yielded a fundamental breakthrough\footnote{The conference version of their paper originally claimed a runtime of $\tilde O_{\epsilon}(n+m)$, but this was subsequently corrected to a runtime of $\tilde O_{\epsilon}((n+m)k)$ in an updated version of the paper that is available online. We validate the necessity of this additional factor $k$ in Section~\ref{sec:lowerbound} by demonstrating that if their algorithm is restricted to a runtime budget of $\tilde{O}_{\epsilon}(n+m)$, the approximation ratio deteriorates to $\tilde O(k^{-1/4})$.} for this problem by achieving an $\tilde O_{\epsilon}((n+m)k)$ time algorithm for approximating the solution within a factor of $1-1/e-\epsilon$. In the years since, numerous efforts have attempted to improve the runtime of this algorithm; however, these works have been successful in only shaving logarithmic factors or improving the dependence on $\epsilon$, leaving the existence of an almost linear-time algorithm as an open question. 

In this work, we resolve this long-standing open question. We present a novel algorithm that approximates the influence maximization problem within a factor of $1-1/e-\epsilon$ in time $\tilde{O}_{\epsilon}(n+m)$, effectively removing the multiplicative dependence on $k$ from the time complexity.

%% file: src/parallel.tex
\section*{Parallel Work and Use of AI}

All algorithms and proofs in this paper were developed entirely by the author. Only after an initial draft of the paper had been written were AI tools (GPT and Claude) used, and only to find typos and grammatical errors and to polish the text. On June 3, 2026, a draft of this result was shared with experts in the community.

In a parallel work, Zhang~\cite{zhang2026budget} obtains the same result, namely a $(1-1/e-\epsilon)$-approximation for influence maximization in time $\tilde O_{\epsilon}(n+m)$. Both works use capped graph traversals as the main ingredient; however, they utilize them in different ways to obtain their algorithms. The author was not aware of the work of Zhang~\cite{zhang2026budget} when starting this research, and based on discussions with the author of~\cite{zhang2026budget}, they were not aware of this work either. According to its AI disclosure, the work of Zhang~\cite{zhang2026budget} was developed through interactions with GPT-6 Astra ultra (released on September 3, 2026). It is not clear whether the ideas of this paper were used in the training of GPT Astra, or whether GPT Astra came up with these ideas independently.

%% file: src/intro.tex
\section{Introduction}

The study of how information, ideas, and behaviors propagate through human populations has evolved from a niche sociological interest into a cornerstone of modern computational social science. In an era defined by hyper-connectivity, the mechanisms of word-of-mouth referral have become more transparent and measurable than ever before. Whether it is the viral spread of a mobile application, the adoption of sustainable farming practices in rural communities, or the dissemination of critical public health information, the underlying structure of social ties acts as a catalyst for large-scale behavioral shifts. The influence maximization problem addresses a fundamental question in this context: if a resource-constrained entity can only influence a small number of initial individuals, denoted by $k$, which specific set of seeds should be targeted to trigger the largest possible cascade of adoption across the entire network? \cite{kempe2003maximizing}.  This question bridges the gap between the structural properties of networks and the dynamic processes that play out over their edges.

The study of influence maximization from a theoretical computer science perspective was pioneered by Kempe, Kleinberg, and Tardos~\cite{kempe2003maximizing}. The problem is typically modeled using a graph $G = (V, E)$ with $n$ vertices and $m$ edges, where each edge is associated with a propagation probability and the goal is to select a seed set of at most $k$ vertices to maximize the spread. They formalized the problem under the independent cascade model and demonstrated that the influence spread function (representing the expected number of influenced individuals) is submodular. By leveraging this structural property, they employed a greedy framework to achieve a $1-1/e-\epsilon$ approximation guarantee in polynomial time. Since this foundational result, research has focused on maintaining the $1-1/e-\epsilon$ approximation guarantee while pushing the runtime down to linear in terms of the graph size. 

We note that the approximation ratio $1-1/e-\epsilon$ is essentially optimal for general directed graphs, as improving upon the $1-1/e$ factor has been shown to be NP-hard~\cite{feige1998threshold,kempe2003maximizing}.  However, the landscape changes for undirected graphs; Khanna and Lucier~\cite{khanna2014influence} demonstrated that in the undirected setting, it is possible to bypass this barrier and obtain approximation ratios strictly better than $1-1/e$. 

A transformative breakthrough was achieved by Borgs et al.~\cite{borgs2014maximizing}, who introduced the framework of reverse reachable (RR) sets. This approach shifted the complexity from traditional simulation-based methods to a sampling-based paradigm, resulting in a running time of $\tilde O_{\epsilon}((n+m)k)$ while maintaining a $1-1/e-\epsilon$ approximation guarantee.  In the decade following this result, many works have refined the runtime of the RR-set framework while  maintaining an approximation factor of $1-1/e-\epsilon$. Notably, TIM/TIM+~\cite{tang2014influence} provided the first practical implementation that maintained the $\tilde O_{\epsilon}((n+m)k)$ time  bound while performing well on real-world graphs. This was followed by other influential frameworks such as IMM~\cite{tang2015influence}, SSA~\cite{nguyen2016stop}, and the sketch-based approach~\cite{cohen2014sketch}, which focused on reducing constant factors, improving stopping conditions, or utilizing alternative sketch-based data structures to estimate influence. Subsequent works further improved the practical performance of this framework, e.g., through online quality assessment in OPIM and OPIM-C~\cite{tang2018online} and through more efficient reverse reachable set generation with tightened bounds~\cite{guo2020influence}. Lakshmanan~\cite{lakshmanan2025influence} studies a variant of reverse influence sampling whose running time is independent of the seed set size, but its guarantee requires $k\epsilon < 1$, which restricts the seed budget to $k < 1/\epsilon$. 
 However, a fundamental barrier remained: a multiplicative dependence on the seed set size $k$ in the runtime. While $k$ is often small in practice, in large-scale network analysis or scenarios where $k$ scales with $n$, this factor represents a significant bottleneck that prevents the algorithm from being truly linear in the size of the graph.

In this work, we resolve this long-standing open question. We present a novel algorithm that approximates the influence maximization problem within a factor of $1-1/e-\epsilon$ in time $\tilde{O}_{\epsilon}(n+m)$, effectively removing the multiplicative dependence on $k$ from the time complexity.

%% file: src/preliminaries.tex
\section{Preliminaries}

In this paper, we formalize the influence maximization problem within its primary setting of social networks and subsequently describe the maximum $k$-coverage formulation to which it is reduced. While the former defines the problem's objective and diffusion dynamics, the latter provides the combinatorial structure necessary for our algorithmic analysis.

\subsection{Influence Maximization in Social Networks}
We represent a social network as a directed graph $G = (V, E)$ with $n = |V|$ nodes and $m = |E|$ edges. Each directed edge $(u, v) \in E$ is associated with a propagation probability $p_{u,v} \in [0, 1]$, representing the likelihood that an activated node $u$ activates node $v$.

The diffusion of influence is governed by the independent cascade model, which operates in discrete time steps. Initially, at $t=0$, a given \textit{seed set} $C \subseteq V$ is activated. At any subsequent step $t \ge 1$, every node $u$ that became active at step $t-1$ is given a single opportunity to activate each of its currently inactive out-neighbors $v$. The attempt succeeds with probability $p_{u,v}$. If successful, $v$ becomes active at step $t$. This process continues until no further activations are possible.

For a seed set $C$, let $\sigma(C)$ denote the expected number of activated nodes at the conclusion of the diffusion process. In the influence maximization problem, we are given the graph $G$, the edge probabilities, and an integer $1 \le k \le n$. Our objective is to find a seed set $C$ of at most $k$ nodes that maximizes the expected spread $\sigma(C)$. For an $0 < \alpha < 1$, we say a solution $C \subseteq V$ is $\alpha$-approximate if $|C| \leq k$ and 
\[ \sigma(C) \geq \alpha \max_{C' \subseteq V, |C'| \le k} \sigma(C'). \]

\subsection{The Maximum $k$-Coverage Problem}
As we demonstrate in Section~\ref{sec:reduction}, the influence maximization problem is closely related to a specific formulation of the maximum $k$-coverage problem.  Let $\mathcal{N} = [n] = \{1, 2, \dots, n\}$ be a finite ground set of elements, and let $\mathcal{M} \subseteq 2^{\mathcal{N}}$ be a family (multiset) of subsets over $\mathcal{N}$. We also denote the size of $\mathcal{M}$ by $m$.\footnote{In the context of the $k$-coverage problem, $n$ and $m$ denote the size of the ground set $\mathcal{N}$ and the size of the multiset $\mathcal{M}$, respectively. This differs from their meaning in the influence maximization setting, where $n$ and $m$ denote the number of nodes and edges of the graph $G$. The intended meaning will be clear from the context.}  For reasons that will become clear later, in this paper, we utilize the dual notion of the maximum $k$-cover problem: we seek to find a subset of at most $k$ elements that maximizes the number of sets in $\mathcal{M}$ containing at least one of these elements. 

In our setting,  multiset $\mathcal{M}$ is not explicitly provided. Instead, the algorithm interacts with an oracle $\mathcal{Q}$ that, upon query, returns a subset $S \in \mathcal{M}$ sampled uniformly and independently from $\mathcal{M}$. Each query to $\mathcal{Q}$ incurs a non-uniform cost equal to the cardinality of the returned set, $|S|$. For a sequence of observed sets $S_1, S_2, \dots, S_x$, the total query cost is defined as $\sum  |S_i|$.

For any subset of elements $C \subseteq \mathcal{N}$, let $F(C)$ denote the \textit{fractional coverage} of $C$, which is the fraction of sets in $\mathcal{M}$ that contain any element of $C$. More precisely, $F$ is formulated as follows:
\[ F(C) = \frac{|\{S \in \mathcal{M} \mid S \cap C \neq \emptyset\}|}{|\mathcal{M}|}. \]
Our goal is to design a randomized algorithm that sequentially queries $\mathcal{Q}$ to find a subset $C$ with $|C| \le k$ such that $F(C) \ge \alpha \cdot \mathsf{opt}$, where $\mathsf{opt}$ is the maximum possible fractional coverage and $\alpha$ is the approximation ratio. We specifically seek an algorithm that provides a $1-1/e-\epsilon$ approximation guarantee while maintaining  almost linear query complexity and almost linear runtime.

Throughout the paper, we assume that every set in $\mathcal{M}$ is non-empty (this holds for reverse reachable sets, which always contain their root). This guarantees that $\mathsf{opt} \ge k/n$ holds in general, since a uniformly random subset of $k$ elements contains any fixed element with probability $k/n$ and hence intersects every set of $\mathcal{M}$ with probability at least $k/n$.

%% file: src/results.tex
\section{Our Results}\label{sec:reduction}
The primary technical contribution of Borgs et al.~\cite{borgs2014maximizing} is a transformative framework that reduces the influence maximization problem to a classic instance of the $k$-cover problem. For a node $u \in V$ and a random realization $R \sim G$ of $G$, Borgs et al.~\cite{borgs2014maximizing} introduced the \textit{reverse reachable set} $S_{R}(u)$ as the subset of nodes of $V$ that can reach $u$ in a realization $R$ of $G$. Specifically, if we were to construct $S_R(u)$ for a random realization $R \sim G$ of $G$, we can initiate a graph traversal (such as BFS or DFS) starting from vertex $u$ and proceeding in the reverse direction of the edges. In this process, each directed edge $(a, b) \in E$ is considered for traversal with probability $p_{a,b}$. The resulting set of visited nodes $S$ contains all nodes from which $u$ is reachable in that specific sampled subgraph. 

The efficacy of the framework proposed by Borgs et al.~\cite{borgs2014maximizing} rests on a fundamental coupling between the diffusion process and the distribution of reverse reachable sets. Specifically, their algorithm is based on the following observation: let $(\mathcal{N}, \mathcal{M})$ be an instance of the maximum $k$-coverage problem where the ground set $\mathcal{N} = \{1, \dots, n\}$ corresponds to the set of nodes $V$, and the multiset $\mathcal{M}$ is formulated in a way that oracle $\mathcal{Q}$ answers queries as follows: 
\begin{enumerate}
	\item Pick a node $v \in V$ uniformly at random.
	\item Generate a reverse reachable set $S(v)$ by sampling edges according to their propagation probabilities and traversing the graph in reverse starting from $v$.
\end{enumerate}

Under this construction, the objective function of the $k$-coverage problem (the fraction of sets in $\mathcal{M}$ that contain at least one element from a seed set $C$) is exactly proportional to the objective function of the influence maximization problem. Formally, for any $C \subseteq V$ we have $\sigma(C) = n \cdot \mathsf{Pr}_{S \sim \mathcal{M}}[C \cap S \neq \emptyset]$.

If the degrees of the nodes in $G$ were constant, the computational time required for the oracle $\mathcal{Q}$ to generate and return a reverse reachable set would be exactly the same as the size of the returned set. Consequently, any algorithm for the $k$-coverage problem with runtime $T(n)$ and query cost $W(n)$  would translate into an algorithm for the influence maximization problem with runtime $T(n) + W(n)$, preserving the same approximation guarantee. In this idealized scenario, achieving both runtime and query cost of $\tilde{O}_{\epsilon}(n)$ for the $k$-coverage instance yields a truly near-linear time solution for influence maximization.

It is not hard to show that this observation holds even when the degrees are not bounded, provided we adjust the parameters of the equivalent $k$-coverage instance. Specifically, the parameter $n$ in the ground set $\mathcal{N}$ would become exactly two times the number of edges of the graph ($2m$) plus the number of isolated nodes.  To formalize this, for each vertex $v \in V$, we can conceptually correspond it with a set of equivalent elements of $\mathcal{N}$ whose count is equal to the degree of $v$ in $G$ (or 1 in case of isolated nodes). In this reformulated instance, the ground set has size $|\mathcal{N}| \leq 2m+n$, and the cost of each oracle query would be equal to the sum of the degrees of the visited nodes which is exactly the time complexity of the underlying BFS or DFS traversal. Accordingly, our objective for the remainder of this paper is to design an algorithm that solves the maximum $k$-coverage problem in our oracle-based setting with a $(1-1/e-\epsilon)$ approximation factor, where both the total runtime and the cumulative query cost are bounded by $\tilde{O}_{\epsilon}(n)$. This would give us an algorithm with runtime $\tilde O_{\epsilon}(n+m)$ for the influence maximization problem that has the same approximation factor.

In Section~\ref{sec:weak}, we adapt the sampling-based framework of Borgs et al.~\cite{borgs2014maximizing} to our specific $k$-coverage setting to achieve a $(1 - 1/e - \epsilon)$ approximation with a query cost and runtime of $\tilde{O}_{\epsilon}(nk)$. While the core algorithmic principles follow their established methodology, we restate the derivation here to ensure the presentation is self-contained and tailored to the nuances of our problem. The approach hinges on constructing a sampled collection $\mathcal{M}'$ that is small enough for efficient processing yet representative enough to preserve coverage guarantees.

Let $\mathsf{opt}$ denote the optimal coverage ratio achievable by any subset of at most $k$ elements in the $k$-coverage problem. Lemma~\ref{lem:concentration} characterizes the sample complexity required to downsample $\mathcal{M}$ into a representative sub-collection $\mathcal{M}'$. Specifically, it ensures that running the standard greedy algorithm on $\mathcal{M}'$ yields a solution that remains near-optimal with respect to the original collection $\mathcal{M}$.

\vspace{0.2cm}
{\noindent \textbf{Lemma}~\ref{lem:concentration} [restated informally]. \textit{Let $\mathcal{M}'$ be a collection of $m' \geq \tilde \Omega(\frac{k}{\epsilon^2 \mathsf{opt}})$ sets sampled uniformly and independently from $\mathcal{M}$. Let $F(C)$ denote the true fraction of sets in $\mathcal{M}$ covered by a fixed subset $C \subseteq \mathcal{N}$ such that $F(C) \le \mathsf{opt}$, and let $F'(C)$ denote the empirical fraction of sets in $\mathcal{M}'$ covered by $C$. For any $0 < \epsilon < 1$, the probability that the empirical coverage deviates from the true coverage by more than $\epsilon \cdot \mathsf{opt}$ is bounded by:
		$$\mathsf{Pr}\big[|F'(C) - F(C)| > \epsilon \cdot \mathsf{opt}\big] \le 2 n^{-3 k}.$$\\}}

We then demonstrate that by iteratively sampling from $\mathcal{M}$ to construct $\mathcal{M}'$ until the cumulative query cost reaches a threshold of $\tilde{O}(nk)$, the number of sampled subsets $m'$ will, with desirable probability, meet or exceed the $\tilde{\Omega}(\frac{k}{\epsilon^2 \mathsf{opt}})$ requirement specified in Lemma~\ref{lem:concentration}. This ensures that the algorithm terminates within the desired complexity bounds while providing a sufficiently large representative sample to guarantee the approximation factor. Ultimately, this approach yields an algorithm with both runtime and sample complexity of $\tilde{O}(nk)$, providing a solution for the $k$-coverage problem that is on par with the runtime and approximation guarantees established by Borgs et al.~\cite{borgs2014maximizing}.

Our first contribution involves a detailed investigation into the necessity of the multiplicative $k$ factor within the sample complexity of the $k$-coverage problem. Let us investigate Lemma~\ref{lem:concentration} to understand this requirement. If we were to ignore the $k$ multiplicative factor in the necessary sample size $m'$, the concentration bound would still yield a deviation probability of roughly $n^{-3}$ for any single fixed subset. While $n^{-3}$ is a significantly small failure probability for one subset, we require this condition to hold simultaneously for all possible subsets of $\mathcal{N}$ of size at most $k$. Because the number of such subsets is exponentially large (around $n^k$) a standard union bound argument requires each individual failure probability to be significantly smaller than the inverse of this cardinality. Consequently, without the additional multiplicative factor $k$ in the sample complexity to sharpen the concentration, the theoretical guarantee for the greedy algorithm would not hold across the entire search space. 

While the existing framework established by Borgs et al.~\cite{borgs2014maximizing} necessitates this multiplicative $k$ factor to guarantee a $(1 - 1/e - \epsilon)$ approximation, it is natural to ask whether a more refined analysis might yield a stronger guarantee for a lower sampling budget. In other words, our first step toward obtaining a linear-time solution is to investigate what specific approximation guarantee is achieved if we employ the algorithm of Borgs et al.~\cite{borgs2014maximizing} but restrict the sampling phase to a budget of only $\tilde{O}_{\epsilon}(n)$. Notably, Borgs et al.~\cite{borgs2014maximizing} already establish that when the sample complexity is reduced by a multiplicative factor of $\beta$, the approximation factor is likewise degraded by a factor of at most $O(\beta)$; this would imply that such an algorithm, with a budget reduction of $\beta = k$, possesses an approximation factor of $\Omega(1/k)$.

In Section~\ref{sec:lowerbound}, we show that for our $k$-coverage problem, any algorithm restricted to a total query cost of $\tilde O(n)$ cannot achieve an expected approximation factor better than $\tilde O(1/\sqrt{k})$. It is important to note that this specific lower bound does not directly carry over to the influence maximization setting of Borgs et al.~\cite{borgs2014maximizing}, as our hard instance construction is not producible within the standard graph-based reachability models. Nonetheless, we construct another example that proves a weaker bound of $\tilde O(k^{-1/4})$ on the approximation factor of the algorithm by Borgs et al.~\cite{borgs2014maximizing} when the budget for DFS time is bounded by $\tilde O(n+m)$.

Our findings in Section~\ref{sec:lowerbound} highlight two key observations. First, in order to approximate the $k$-coverage problem within a factor of $(1 - 1/e - \epsilon)$ while maintaining almost linear query cost complexity, one must go beyond the standard oracle $\mathcal{Q}$ queries that return independent samples. Second, to achieve a $(1 - 1/e - \epsilon)$ approximation for the influence maximization problem in social networks, it is necessary to go beyond the sampling framework established by Borgs et al.~\cite{borgs2014maximizing}, as a simple reduction in sample size to achieve linear time complexity fundamentally compromises the approximation guarantee. In Section~\ref{sec:strong}, we introduce a stronger type of oracle $\mathcal{Q}_s$ that enables us to approximate the $k$-cover problem within a factor of $(1 - 1/e - \epsilon)$ while keeping both the runtime and query cost bounded by $\tilde{O}_{\epsilon}(n)$. We then demonstrate in Section~\ref{sec:practical_queries} that such an oracle can be embedded into a reduction from influence maximization to $k$-coverage. This embedding makes it possible to approximate the influence maximization problem for social networks within a factor of $(1 - 1/e - \epsilon)$ in total time $\tilde{O}_{\epsilon}(n+m)$, where $n$ and $m$ are the number of nodes and edges in the network, respectively.

To bypass the limitations established in the linear-cost regime, we introduce a stronger query model. Oracle $\mathcal{Q}_s(A)$ takes a subset of elements $A \subseteq \mathcal{N}$ as input and returns a random set $S \in \mathcal{M}$ sampled uniformly and independently from the sub-collection of sets that are \textit{not} covered by any element in $A$. The cost of this query remains equal to the cardinality of the returned set. In our algorithm, we construct the solution in multiple phases to maintain efficiency. In each phase, we take our current partial solution as the input $A$, which is fed to $\mathcal{Q}_s$ so that it gives us a representative sample of the residual problem. We utilize this capability to develop a phased greedy algorithm that achieves a $(1 - 1/e - \epsilon)$ approximation for the $k$-coverage problem while maintaining both runtime and query complexity bounded by $\tilde{O}_{\epsilon}(n)$. The algorithm operates by greedily selecting elements from the sampled collection $\mathcal{M}'$ until their marginal gains decay below a specific threshold, indicating that the current sample is no longer representative of the remaining optimal coverage. By refreshing the sample through the new oracle in the beginning of each of the $\tilde O_{\epsilon}(1)$ phases,  we show that our new algorithm is able to use $\mathcal{Q}_s$ to achieve an approximation factor of $1-1/e-\epsilon$ with almost linear query cost and runtime. This is the most technically involved part of our contribution.

Finally, in Section~\ref{sec:practical_queries}, we demonstrate how the adaptive oracle $\mathcal{Q}_s$ can be realized within the context of social network diffusion to finalize our reduction from influence maximization to $k$-coverage. Although $\mathcal{Q}_s$ is strong enough to solve the $k$-coverage problem, a standard graph traversal (BFS/DFS) cannot directly implement it by only incurring a cost proportional to the size of the returned set, as a naive rejection strategy might spend excessive time exploring large, invalid components. To address this, we introduce a middle-ground oracle $\mathcal{Q}_{\ell}$, which limits the exploration cost for any single sample to a threshold $\ell$. We prove in Lemma~\ref{lem:threshold} that there always exists a threshold enabling us to simulate $\mathcal{Q}_s$ effectively, thereby allowing us to plug our $k$-coverage framework into the influence maximization problem to achieve a $(1-1/e-\epsilon)$ approximation in $\tilde{O}_{\epsilon}(n+m)$ time.

\begin{theorem}[restatement of Corollary~\ref{cor:final}]
	The influence maximization problem can be approximated within a factor of $1-1/e-\epsilon$ in time $\tilde O_{\epsilon}(n+m)$.
\end{theorem}

%% file: src/nk.tex
\section{A $\tilde O_{\epsilon}(nk)$ Time Solution for $k$-cover}\label{sec:weak}
In this section, we present an algorithm that achieves a $(1 - 1/e - \epsilon)$ approximation guarantee for our $k$-cover problem while maintaining a query cost and runtime of $\tilde{O}_{\epsilon}(nk)$.  Let $\mathsf{opt}$ denote the maximum fraction of the subsets of $\mathcal{M}$ that can be covered by an optimal choice of at most $k$ elements from $\mathcal{N}$. The core intuition behind our approach relies on random sampling. For any fixed subset of elements $C \subseteq \mathcal{N}$, if we sample roughly $O(\frac{\log n}{\epsilon^2 \mathsf{opt}})$ subsets from $\mathcal{M}$ to form a sampled collection $\mathcal{M}'$, the fraction of subsets in $\mathcal{M}'$ that get covered by $C$ is highly concentrated around its true expected value in $\mathcal{M}$. Specifically, the empirical coverage fraction will be close to the true coverage fraction with an additive error of at most $\epsilon \cdot \mathsf{opt}$. 

However, our algorithm must evaluate and compare many different subsets. To guarantee that this additive error bound holds simultaneously for every subset $C \subseteq \mathcal{N}$ of size at most $k$, we must apply a union bound. Since the number of such subsets is roughly $O(n^k)$, the failure probability for a single subset must be significantly smaller than $n^{-k}$. This requirement introduces an extra multiplicative $k$ factor into the necessary sample size, leading us to sample $m' \geq \Theta(\frac{k \log n}{\epsilon^2 \mathsf{opt}})$ sets. Once we have this sampled collection $\mathcal{M}'$, we simply run the standard greedy algorithm on it. Therefore, our algorithm proceeds as follows:
\begin{enumerate}
	\item Query the oracle $\mathcal{Q}$ many times to construct a sampled set $\mathcal{M}'$ from $\mathcal{M}$ whose total sum of set sizes is at least $16kn \log n / \epsilon ^ 2$.
	\item Initialize an empty set $\mathcal{C} = \emptyset$. For $i = 1$ to $k$, find an element $e \in \mathcal{N} \setminus \mathcal{C}$ that covers the maximum number of previously uncovered sets in $\mathcal{M}'$ and add $e$ to $\mathcal{C}$.
\end{enumerate}

To prove the correctness of our approach, we first establish Lemma~\ref{lem:concentration} that shows that the fractional coverage of any subset of size at most $k$ is preserved in the sample $\mathcal{M}'$ with high probability. 

\begin{lemma} \label{lem:concentration}
	Let $\mathcal{M}'$ be a collection of $m' \geq \frac{8 k \log n}{\epsilon^2 \mathsf{opt}}$ sets sampled uniformly and independently from $\mathcal{M}$. Let $F(C)$ denote the true fraction of sets in $\mathcal{M}$ covered by a fixed subset $C \subseteq \mathcal{N}$ such that $F(C) \le \mathsf{opt}$, and let $F'(C)$ denote the empirical fraction of sets in $\mathcal{M}'$ covered by $C$. For any $0 < \epsilon < 1$, the probability that the empirical coverage deviates from the true coverage by more than $\epsilon \cdot \mathsf{opt}$ is bounded by:
	$$\mathsf{Pr}\big[|F'(C) - F(C)| > \epsilon \cdot \mathsf{opt}\big] \le 2 n^{-3 k}.$$
\end{lemma}

\begin{proof}
	Consider a fixed subset $C \subseteq \mathcal{N}$ where $F(C) \le \mathsf{opt}$. Let $x_1, \dots, x_{m'}$ be independent Bernoulli indicator variables where $x_i = 1$ if the $i$-th sampled set is covered by $C$, and $0$ otherwise. The empirical average is $F'(C) = \frac{1}{m'} \sum x_i$, with expectation $\mathbb{E}[x_i] = F(C)$ and variance $\sigma^2 = F(C)(1-F(C)) \le F(C)$.
	
	Bernstein's inequality~\cite{bernstein1924modification} states that for the average of independent random variables bounded by $r$, with variance $\sigma^2$, the probability of a deviation $\delta$ is bounded by
	$$\mathsf{Pr}\big[|F'(C) - F(C)| > \delta\big] \le 2 \exp \left( - \frac{m' \delta^2}{2\sigma^2 + \frac{2}{3}r\delta} \right).$$
	
	In our setting, we have $r=1$ and $\delta = \epsilon \cdot \mathsf{opt}$. We substitute $m' \geq \frac{8 k \log n}{\epsilon^2 \mathsf{opt}}$ and apply the bounds $\sigma^2 \le F(C) \le \mathsf{opt}$ to simplify the exponent step by step:
	\begin{align*}
		\frac{m' \delta^2}{2\sigma^2 + \frac{2}{3}r\delta} &\ge \frac{\left( \frac{8 k \log n}{\epsilon^2 \mathsf{opt}} \right) (\epsilon \cdot \mathsf{opt})^2}{2 F(C) + \frac{2}{3} (1) (\epsilon \cdot \mathsf{opt})} \\
		&= \frac{8 k \mathsf{opt} \log n}{2 F(C) + \frac{2}{3} \epsilon \mathsf{opt}} \\
		&\ge \frac{8 k \mathsf{opt} \log n}{2 \mathsf{opt} + \frac{2}{3} \mathsf{opt}} \\
		&= \frac{8 k \mathsf{opt} \log n}{\frac{8}{3} \mathsf{opt}} \\
		&= 3 k \log n.
	\end{align*}
	
	Substituting this back into the original inequality, we obtain
	$$\mathsf{Pr}\big[|F'(C) - F(C)| > \epsilon \cdot \mathsf{opt}\big] \le 2 \exp(-3 k \log n) \le 2 \exp(-3 k \ln n) = 2 n^{-3k}$$
	which completes the proof for any fixed set $C$ whose coverage is no greater than the optimal coverage.
\end{proof}

We now formally define Algorithm~\ref{alg:simple}. To ensure the query cost remains $O(nk \log n / \epsilon ^ 2)$, the algorithm samples sets from the oracle $\mathcal{Q}$ until the cumulative cost of the sets reaches $16nk \log n / \epsilon ^ 2$.

\begin{algorithm}[H]\label{alg:simple}
	\DontPrintSemicolon
	\KwIn{Oracle $\mathcal{Q}$, element set $\mathcal{N} = [n]$, budget $k$, error parameter $\epsilon$}
	\KwOut{Subset $\mathcal{C} \subseteq \mathcal{N}$ such that $|\mathcal{C}| \le k$}
	$\mathcal{C} \gets \emptyset$\;
	$\mathcal{M}' \gets \emptyset$\;
	$q \gets 0$\;
	\While{$q < 16nk \log n / \epsilon ^ 2$}{
		$S \gets \text{Query } \mathcal{Q}$\;
		$\mathcal{M}' \gets \mathcal{M}' \cup \{S\}$\;
		$q \gets q + |S|$\;
	}
	\For{$i = 1$ \KwTo $k$}{
		$u \gets \arg\max_{e \in \mathcal{N}} |\{S \in \mathcal{M}' \mid e \in S \text{ and }S \cap \mathcal{C} = \emptyset\}|$\;
		$\mathcal{C} \gets \mathcal{C} \cup \{u\}$\;
	}
	\Return{$\mathcal{C}$}\;
	\caption{\textsf{Simple Greedy}}
\end{algorithm}

Before proving our main theorem, we state Lemma~\ref{lem:sample_size} that gives a lower bound on the value of $m'$ (the number of sampled sets from $\mathcal{M}$). We apply this lower bound to Lemma~\ref{lem:concentration} in the proof of Theorem~\ref{thm:main_result}.
\begin{lemma} \label{lem:sample_size}
	With probability at least $1/2$, Algorithm~\ref{alg:simple} samples $m' \geq \frac{8 k \log n}{\epsilon^2 \mathsf{opt}}$ sets from $\mathcal{M}$ before reaching the cumulative cost threshold.
\end{lemma}

\begin{proof}
	First, let $e^* \in [n]$ be an element such that its coverage fraction $F(\{e^*\})$ is maximized. Since any optimal solution of size $k \ge 1$ must cover at least as many sets as the single best element, it immediately follows that $F(\{e^*\}) \le \mathsf{opt}$.
	
	Next, we bound the average size of the sets in $\mathcal{M}$ using a double-counting argument. The sum of the sizes of all sets in $\mathcal{M}$ is exactly equal to the total number of element-set incidences. Let $m$ be the total number of sets in $\mathcal{M}$. The average set size is therefore:
	$$\frac{1}{m} \sum_{S \in \mathcal{M}} |S| = \frac{1}{m} \sum_{e \in \mathcal{N}} \big(m \cdot F(\{e\})\big) = \sum_{e \in \mathcal{N}} F(\{e\}).$$
	Since $F(\{e\}) \le F(\{e^*\}) \le \mathsf{opt}$ for every element, we can bound the sum over all $n$ elements by:
	$$\sum_{e \in \mathcal{N}} F(\{e\}) \le n \cdot F(\{e^*\}) \le n \cdot \mathsf{opt}.$$
	
	Now, for the sake of analysis, assume Algorithm~\ref{alg:simple} continues sampling uniformly and independently even if the cumulative size of the sets goes beyond the stopping threshold. Let us look at the first $m^* = \frac{8 k \log n}{\epsilon^2 \mathsf{opt}}$ random selections of sets used to construct $\mathcal{M}'$. 
	
	By linearity of expectation, the expected total sum of the sizes of these $m^*$ sets is bounded by:
	$$\mathbb{E}\big[ \text{total size of the first } m^* \text{ sets} \big] \le m^* \cdot (n \cdot \mathsf{opt}) = \left( \frac{8 k \log n}{\epsilon^2 \mathsf{opt}} \right) \cdot n \cdot \mathsf{opt} = \frac{8 n k \log n}{\epsilon^2}.$$
	
	Since set sizes are non-negative random variables, we can apply Markov's inequality. The probability that the sum of the sizes of these $m^*$ sets exceeds twice its expected value is at most $1/2$. Thus, with probability at least $1/2$, the total size of the first $m^*$ sets is bounded by:
	$$2 \cdot \frac{8 n k \log n}{\epsilon^2} = \frac{16 n k \log n}{\epsilon^2}.$$
	
	Because $\frac{16 n k \log n}{\epsilon^2}$ is exactly the stopping threshold of our algorithm, achieving a cumulative size below this bound implies that with probability at least $1/2$, Algorithm~\ref{alg:simple} successfully queries at least $m^* = \frac{8 k \log n}{\epsilon^2 \mathsf{opt}}$ sets before the while-loop terminates.
\end{proof}

We are now ready to prove Theorem~\ref{thm:main_result}.

\begin{theorem}[restatement of the result of Borgs et al.~\cite{borgs2014maximizing} ] \label{thm:main_result}
	For any $n \geq 4$, $k \geq 1$, and  $0 < \epsilon < 1$, with probability at least $1/3$, Algorithm~\ref{alg:simple} achieves an approximation factor of $(1 - 1/e - 2\epsilon)$ for the maximum $k$-cover problem. Both the query cost and the runtime of the algorithm are bounded by $O(nk \log n / \epsilon ^ 2)$.
\end{theorem}

\begin{proof}
   By Lemma~\ref{lem:sample_size}, with probability at least $1/2$, the number of sampled sets in $\mathcal{M}'$ (denoted by $m'$) is at least $\frac{8 k \log n}{\epsilon^2 \mathsf{opt}}$. Let us ignore the failure probability for now. To bound the error of the algorithm, we apply the union bound on top the guarantee of Lemma~\ref{lem:concentration}. Let $\mathcal{P}$ be the collection of all subsets $C \subseteq \mathcal{N}$ such that $|C| \le k$. $F(C) \le \mathsf{opt}$ holds for all such subsets by optimality of $\mathsf{opt}$. The number of subsets of size at most $k$ is $\sum_{i=1}^k \binom{n}{i} \le n^k$. Using Lemma \ref{lem:concentration}, the probability that there exists any $C \in \mathcal{P}$ such that $|F'(C) - F(C)| > \epsilon \cdot \mathsf{opt}$ is bounded by
	$$ \mathsf{Pr}\big[\exists C \in \mathcal{P} : |F'(C) - F(C)| > \epsilon \cdot \mathsf{opt}\big] \le n^k \cdot 2n^{-3k} = 2n^{-2k} .$$
	Thus, with  probability at least $1 - 2n^{-2k}$, for every relevant subset $C \in \mathcal{P}$, the empirical coverage $F'(C)$ is within $\epsilon \cdot \mathsf{opt}$ of the true coverage $F(C)$. Adding the $1/2$ failure probability of  $m'$ not reaching the desired threshold, makes the total failure probability bounded by $2n^{-2k} + 1/2 \leq 1/2 + 2 (4^{-2}) < 2/3$. This implies that the following conditions hold with probability at least $1/3$ which is the desired success probability of the statement of the theorem.
	
	Let $C^*$ be the optimal set of $k$ elements such that $F(C^*) = \mathsf{opt}$. On the sampled collection $\mathcal{M}'$, the greedy algorithm finds a solution $\mathcal{C}$ such that $F'(\mathcal{C}) \ge (1 - 1/e) F'(C^*)$. Applying the concentration bound to both $C^*$ and $\mathcal{C}$ we get
	\begin{align*}
		F(\mathcal{C}) &\ge F'(\mathcal{C}) - \epsilon \cdot \mathsf{opt} \\
		&\geq (1 - 1/e) F'(C^*) - \epsilon \cdot \mathsf{opt} \\
		&\geq (1 - 1/e) (F(C^*) - \epsilon \cdot \mathsf{opt}) - \epsilon \cdot \mathsf{opt} \\
		&= (1 - 1/e)\mathsf{opt} - (1 - 1/e)\epsilon \mathsf{opt} - \epsilon \mathsf{opt} \\
		&\geq (1 - 1/e)\mathsf{opt} - \epsilon \mathsf{opt} - \epsilon \mathsf{opt} \\
		&= (1 - 1/e - 2\epsilon)\mathsf{opt}
	\end{align*}
	where we use the fact that $(1 - 1/e) < 1$.
	
	The total query cost is bounded by $O(\frac{nk \log n}{\epsilon^2})$ and the optimal element in each iteration of the greedy algorithm can be found in time $O(n)$ which takes time $O(nk)$ in total.
\end{proof}

%% file: src/curse.tex
\section{The Lightweight Greedy Algorithm and the Lower Bound}\label{sec:lowerbound}
While Algorithm~\ref{alg:simple} provides a strong $(1 - 1/e - \epsilon)$ approximation guarantee, its query cost and runtime both scale as $\tilde{O}_{\epsilon}(nk)$. In many large-scale applications, particularly where $n$ and $k$ are both large, this multiplicative factor of $k$ is undesirable. Ideally, we seek an algorithm that runs in time almost linear in $n$. 

The $k$ factor in our previous analysis arises as a mechanical necessity to satisfy the union bound over all $\binom{n}{k} \approx n^k$ possible subsets. To bypass this, one might consider a lightweight version of the greedy approach where the sampling threshold is reduced by a factor of $k$. Furthermore, the $O(nk)$ runtime of the greedy selection step can be improved to $O(n + q' \log n)$ where $q'$ is the sum of the set sizes in $\mathcal{M}'$ by using a priority queue (heap) to maintain the marginal gains of each element. We formalize this in Algorithm~\ref{alg:fast_greedy}.

\begin{algorithm}[H]\label{alg:fast_greedy}
	\DontPrintSemicolon
	\KwIn{Oracle $\mathcal{Q}$, element set $\mathcal{N} = [n]$, budget $k$, error parameter $\epsilon$}
	\KwOut{Subset $\mathcal{C} \subseteq \mathcal{N}$ such that $|\mathcal{C}| \le k$}
	$\mathcal{C} \gets \emptyset, \mathcal{M}' \gets \emptyset, q \gets 0$\;
	\While{$q < 16n \log n / \epsilon ^ 2$}{
		$S \gets \text{Query } \mathcal{Q}$\;
		$\mathcal{M}' \gets \mathcal{M}' \cup \{S\}$\;
		$q \gets q + |S|$\;
	}
	Build a max-heap $H$ containing all $e \in \mathcal{N}$ with keys $w_e = |\{S \in \mathcal{M}' \mid e \in S\}|$\;
	\For{$i = 1$ \KwTo $k$}{
		$u \gets H.\text{extract\_max}()$\;
		$\mathcal{C} \gets \mathcal{C} \cup \{u\}$\;
		Update keys in $H$ for elements that overlap with $u$ on sets in $\mathcal{M}'$\;
	}
	\Return{$\mathcal{C}$}\;
	\caption{\textsf{The Lightweight Greedy Algorithm}}
\end{algorithm}

The query cost of Algorithm~\ref{alg:fast_greedy} is $\tilde{O}_{\epsilon}(n)$, a $k$-factor improvement over Algorithm~\ref{alg:simple}. However, by reducing the sample size, we lose the uniform concentration guarantee over all subsets of size $k$. This raises a natural question: how much does the approximation factor deteriorate in this case? Borgs et al.~\cite{borgs2014maximizing} already establish that when the sample complexity is reduced by a multiplicative factor of $\beta$, the approximation factor is likewise degraded by a factor of at most $O(\beta)$; this would imply that Algorithm~\ref{alg:fast_greedy}, with a budget reduction of $\beta = k$, possesses an approximation factor of $\Omega(1/k)$.

The $\Omega(1/k)$ approximation guarantee is indeed a trivial observation. By the property of submodular diminishing returns, the best single element always covers at least a $1/k$ fraction of the sets covered by the optimal $k$ elements, i.e., $\max_e F(\{e\}) \ge \frac{1}{k} \mathsf{opt}$. Also, one can think of Algorithm~\ref{alg:fast_greedy} as an algorithm which is $1-1/e-\epsilon$ approximation for $k=1$ and moreover, it adds $k-1$ additional elements to this solution. Thus, its final coverage cannot be worse than $(1-1/e-\epsilon)\mathsf{opt}_1$ where $\mathsf{opt}_1 = \max_{e \in \mathcal{N}} F(\{e\})$. Therefore, Algorithm~\ref{alg:fast_greedy} is at least $\frac{1-1/e-\epsilon}{k}$ approximation.

The gap between the $\tilde O_{\epsilon}(nk)$ complexity of Algorithm~\ref{alg:simple}, which guarantees a $(1-1/e-\epsilon)$ approximation, and the $\tilde O_{\epsilon}(n)$ complexity of Algorithm~\ref{alg:fast_greedy} raises a fundamental question: can we achieve a $(1-1/e-\epsilon)$ approximation with $\tilde{O}_{\epsilon}(n)$ query cost, or is the $k$ multiplicative factor in the query cost necessary?

We now establish a strict lower bound demonstrating that any algorithm restricted to an $O(n)$ query cost cannot achieve an approximation factor better than $O(1/\sqrt{k})$. In other words, the algorithm's performance deteriorates by a gap of $\Omega(\sqrt{k})$ compared to the optimal solution. The intuition is illustrated in Figure~\ref{fig:lower_bound_construction}. If an algorithm is constrained to an $O(n)$ budget, it cannot afford to observe the entire instance. We formalize this limitation in Theorem~\ref{lem:lower_bound}.

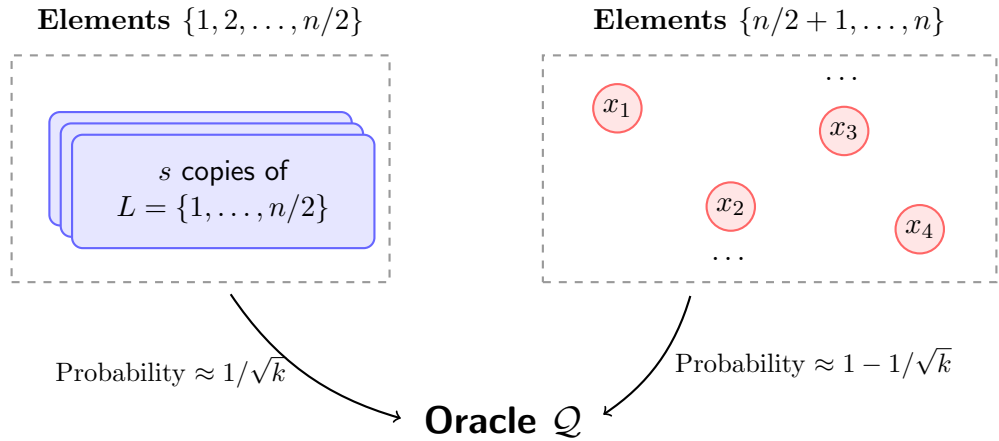
\begin{figure}[htbp]
	\centering
	\begin{tikzpicture}[
		font=\sffamily,
		largeSet/.style={rectangle, rounded corners, draw=blue!60, fill=blue!10, thick, minimum width=4cm, minimum height=1.5cm, align=center},
		singleton/.style={circle, draw=red!60, fill=red!10, thick, inner sep=2pt, minimum size=0.6cm},
		domainBox/.style={rectangle, draw=gray!80, dashed, thick, inner sep=10pt}
		]
		
		\node[domainBox, minimum width=5cm, minimum height=3cm] (domain1) {};
		\node[above=0.1cm of domain1.north, font=\bfseries] {Elements $\{1, 2, \dots, n/2\}$};
		
		\node[largeSet] (L1) at (domain1.center) {};
		\node[largeSet, yshift=-0.15cm, xshift=0.15cm] (L2) at (domain1.center) {};
		\node[largeSet, yshift=-0.3cm, xshift=0.3cm] (L3) at (domain1.center) {$s$ copies of \\ $L = \{1, \dots, n/2\}$};
		
		\node[domainBox, right=2cm of domain1, minimum width=6cm, minimum height=3cm] (domain2) {};
		\node[above=0.1cm of domain2.north, font=\bfseries] {Elements $\{n/2+1, \dots, n\}$};
		
		\node[singleton] (s1) at ([xshift=1cm, yshift=0.8cm]domain2.west) {$x_1$};
		\node[singleton] (s2) at ([xshift=2.5cm, yshift=-0.5cm]domain2.west) {$x_2$};
		\node[singleton] (s3) at ([xshift=4cm, yshift=0.5cm]domain2.west) {$x_3$};
		\node[singleton] (s4) at ([xshift=5cm, yshift=-0.8cm]domain2.west) {$x_4$};
		
		\node[below=0.2cm of s2] {$\dots$};
		\node[above=0.2cm of s3] {$\dots$};
		
		
		\node[below=1.5cm of domain1, xshift=4cm] (oracle) {\Large \textbf{Oracle $\mathcal{Q}$}};
		\draw[->, thick, shorten >=5pt, shorten <=5pt] (domain1.south) ++(0.3,0) to[bend right=20] node[midway, left=0.1cm, font=\small] {Probability $\approx 1/\sqrt{k}$} (oracle.west);
		\draw[->, thick, shorten >=5pt, shorten <=5pt] (domain2.south) ++(-1,0) to[bend left=20] node[midway, right=0.1cm, font=\small] {Probability $\approx 1-1/\sqrt{k}$} (oracle.east);
		
	\end{tikzpicture}
	\caption{Construction of the hard instance $(\mathcal{N},\mathcal{M})$.}
	\label{fig:lower_bound_construction}
\end{figure}

\begin{theorem}\label{lem:lower_bound}
	For any $k \geq 64$ and any $1 \leq c \leq \sqrt{k}$ there exists an instance $(\mathcal{N},\mathcal{M})$ of the $k$-cover problem such that any algorithm that makes queries to oracle $\mathcal{Q}$ with a total cost of $c\cdot n$ cannot achieve an expected approximation factor better than $O(c/\sqrt{k})$ for the  $k$-Cover problem.
\end{theorem}

\begin{proof}
	Let $s = \lceil \sqrt{k} \rceil$ and set $n = 2(k-s)s$. We construct a ground set $\mathcal{N}$ of $n$ elements partitioned into two blocks: $B_1 = \{1, 2, \dots, n/2\}$ and $B_2 = \{n/2+1, \dots, n\}$. 
	
	We construct the target multiset $\mathcal{M}$ of size $k$ as follows:
	\begin{enumerate}
		\item We include $s$ identical copies of a large set $L = \{1, 2, \dots, n/2\}$ in $\mathcal{M}$. Note that the cardinality of $L$ is $n/2$.
		\item We select a subset $Q \subset B_2$ uniformly at random such that $|Q| = k - s$. For each $x \in Q$, we include a singleton set $S_x = \{x\}$ in $\mathcal{M}$. 
	\end{enumerate}
	
	 First observe that an optimal choice of at most $k$ elements would select exactly one element from $B_1$ and all $k-s$ elements of $Q$. This solution has a size of $1 + k - s \le k$ and covers all $s$ copies of $L$ as well as all $k-s$ singleton sets. Thus, the optimal coverage is exactly $k$ sets (or $\mathsf{opt} = 1$ in fractional terms).
	
	Now, consider an arbitrary algorithm that interacts with oracle $\mathcal{Q}$ and is restricted to a total query cost of $cn$. Since the cost of each query that results in set $L$ is $n/2$, our oracle returns at most $2c$ such sets throughout the lifetime of the algorithm. Moreover,  for each query, the oracle returns $L$ with probability $s/k$ and a singleton set with probability $1-s/k$.  \claudeedit{All of these queries happen before the $(2c+1)$-th query that returns $L$ (since that query would exceed the budget), and the expected index of that query is $(2c+1)k/s$.} Thus, the expected number \claudeedit{$\mathbb{E}[V]$} of singleton sets that our algorithm encounters throughout its lifetime is bounded by $\claudeedit{(2c+1)k/s \leq (2c+1)\sqrt{k}}$. Notice that our construction of $Q$ is random and thus our algorithm is indifferent between the singleton sets that it does not visit during its lifetime. \claudeedit{If $V \le n/4$, each of the at most $k$ elements that our algorithm picks among the at least $n/2 - V \ge n/4$ unvisited elements of $B_2$ belongs to $Q$ with probability at most $(k-s)/(n/4)$. Otherwise, our algorithm covers at most $k-s$ singleton sets, and by Markov's inequality $\Pr[V > n/4] \le 4\mathbb{E}[V]/n$.} Thus, the expected number of singleton sets that are covered by the solution of our algorithm would be bounded by 
\claudeedit{
\begin{align*}
	\mathbb{E}[V] + k\frac{k-s}{n/4} + (k-s)\frac{4\mathbb{E}[V]}{n} &= \mathbb{E}[V] + \frac{2k}{s} + \frac{2\mathbb{E}[V]}{s} && (\text{since } n = 2(k-s)s)\\
	&\leq (2c+1)\sqrt{k} + 2\sqrt{k} + 2(2c+1) \\
	&\leq (3c+4)\sqrt{k} && (\text{since } \sqrt{k} \geq 8).
\end{align*}
}
	This in addition to the $s$ large sets that our algorithm may cover results in a total coverage $\claudeedit{(3c+4)\sqrt{k} + s \leq (3c+5)\sqrt{k} + 1 \leq (3c+6)\sqrt{k}}$ and thus the coverage ratio is bounded by $\claudeedit{(3c+6)/\sqrt{k}}$.
\end{proof}

Although the example illustrated in Theorem~\ref{lem:lower_bound} does not directly correspond to an instance of the influence maximization problem, we show a similar example in that context that shows Algorithm~\ref{alg:fast_greedy} with a total DFS/BFS cost of $c \cdot (m+n)$, does not provide an approximation guarantee better than $O(c k^{-1/4})$.

\begin{theorem}\label{lem:lower_bound2}
	For any $n \geq k \geq 2$ such that $k < n^{4/5}$ and any $1 \leq c \leq k^{1/4}$ there exists an instance of the influence maximization problem with $n$ nodes for which the approximation factor of Algorithm~\ref{alg:fast_greedy} with total DFS/BFS budget of $c \cdot (m+n)$ has an expected approximation factor of $O(c \cdot k^{-1/4})$.
\end{theorem}
\begin{proof}
	Let $\alpha = 1-\gamma/2$, $\beta = 1-5\gamma/4$, and let $\gamma$ be chosen such that $k = 1+n^\gamma$. Assume for simplicity that $n^{\alpha}$, $n^\beta$, and $n^{\gamma}$ are all integer numbers. In our influence maximization instance the probabilities for all edges are equal to 1. Moreover, our graph consists of one directed cycle of size $n^{\alpha}$ and $n^\gamma$ directed cycles of size $n^{\beta}$. The remaining $n- n^{\alpha} - n^{\beta+\gamma}$ vertices are isolated nodes. The nodes are labelled from 1 to $n$ in a random order. Since the out-degree of each vertex in our graph is either 0 or 1, the number of edges of our graph is bounded by the number of vertices and thus $n+m \leq 2n$.
	
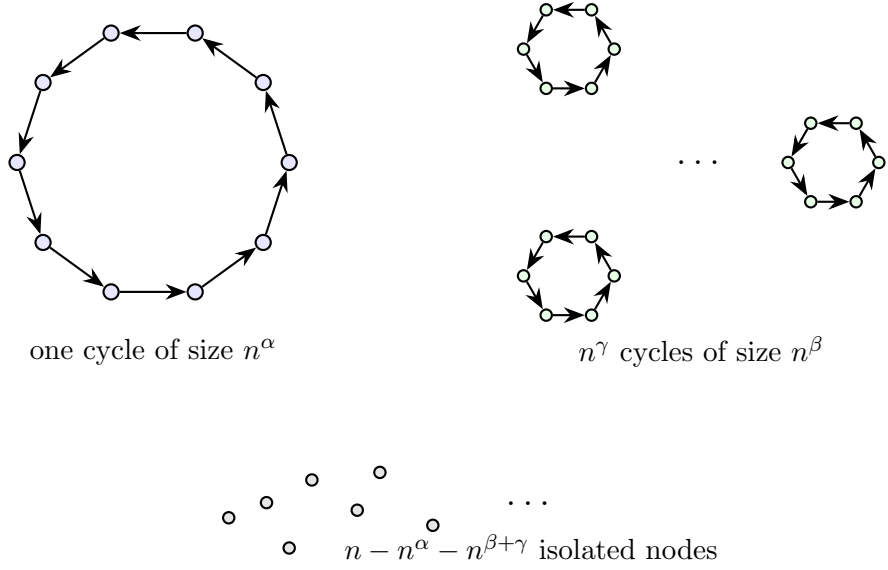
\begin{figure}[h]
	\centering
	\usetikzlibrary{arrows.meta}
	\begin{tikzpicture}[
		main node/.style={circle, draw, fill=blue!10, minimum size=0.2cm, inner sep=0pt},
		small node/.style={circle, draw, fill=green!10, minimum size=0.15cm, inner sep=0pt},
		isolated node/.style={circle, draw, fill=gray!20, minimum size=0.15cm, inner sep=0pt},
		thick,
		>={Stealth[scale=1.2]}
		]
		
		\begin{scope}[xshift=0cm, yshift=0cm]
			\foreach \i in {1,...,10} {
				\node[main node] (L\i) at ({360/10 * (\i-1)}:1.8cm) {};
			}
			\foreach \i [evaluate=\i as \nextnode using {int(mod(\i,10)+1)}] in {1,...,10} {
				\draw[->] (L\i) -- (L\nextnode);
			}
			\node at (0, -2.5) {one cycle of size $n^\alpha$};
		\end{scope}
		
		\begin{scope}[xshift=5.5cm, yshift=1.5cm]
			\foreach \i in {1,...,6} {
				\node[small node] (S1\i) at ({360/6 * (\i-1)}:0.6cm) {};
			}
			\foreach \i [evaluate=\i as \nextnode using {int(mod(\i,6)+1)}] in {1,...,6} {
				\draw[->] (S1\i) -- (S1\nextnode);
			}
		\end{scope}
		
		\begin{scope}[xshift=5.5cm, yshift=-1.5cm]
			\foreach \i in {1,...,6} {
				\node[small node] (S2\i) at ({360/6 * (\i-1)}:0.6cm) {};
			}
			\foreach \i [evaluate=\i as \nextnode using {int(mod(\i,6)+1)}] in {1,...,6} {
				\draw[->] (S2\i) -- (S2\nextnode);
			}
		\end{scope}
		
		\begin{scope}[xshift=9cm, yshift=0cm]
			\foreach \i in {1,...,6} {
				\node[small node] (S3\i) at ({360/6 * (\i-1)}:0.6cm) {};
			}
			\foreach \i [evaluate=\i as \nextnode using {int(mod(\i,6)+1)}] in {1,...,6} {
				\draw[->] (S3\i) -- (S3\nextnode);
			}
		\end{scope}
		
		\node at (7.25, 0) {\Large $\dots$};
		\node at (7.25, -2.5) {$n^\gamma$ cycles of size $n^\beta$};
		
		\begin{scope}[xshift=1.5cm, yshift=-4.5cm]
			\node[isolated node] at (0,0) {};
			\node[isolated node] at (0.6,0.3) {};
			\node[isolated node] at (-0.5,-0.2) {};
			\node[isolated node] at (1.2,-0.1) {};
			\node[isolated node] at (0.3,-0.6) {};
			\node[isolated node] at (1.5,0.4) {};
			\node[isolated node] at (2.2,-0.3) {};
			
			\node at (3.5, 0) {\Large $\dots$};
			\node at (3.5, -0.6) {$n - n^\alpha - n^{\beta+\gamma}$ isolated nodes};
		\end{scope}
		
	\end{tikzpicture}
	\caption{The graph construction for the lower bound proof represented as a directed graph.}
	\label{fig:graph_construction}
\end{figure}
	
	Each time we run a DFS on any of the vertices of the cycle with $n^{\alpha}$ nodes we spend a DFS budget of $n^\alpha$. Since the total DFS budget is bounded by $c \cdot (n+m) \leq 2c \cdot n$ and the odds that we choose a node of that cycle to start our DFS is $n^{\alpha-1}$,  the expected number of times we run a DFS from a randomly selected node is bounded by $2c \cdot n^{2-2\alpha}$. Among these initial nodes, the expected number of times we start our DFS from a node that is in a cycle of size $n^{\beta}$ is $2c \cdot n^{2-2\alpha - (1-\beta-\gamma)} = 2c \cdot n^{1-2\alpha + \beta + \gamma}$.
	
	Since $k = 1+n^{\gamma}$, the optimal solution is to choose one node from each of the cycles. Such a seed set would activate $n^\alpha + n^{\beta + \gamma}$ nodes. Our algorithm however, encounters at most $2c \cdot n^{1-2\alpha + \beta + \gamma}$ cycles of size $n^{\beta}$ in expectation and has no information about which of the unvisited nodes belong to such cycles. Thus, the best way for our algorithm to maximize the coverage is to
	\begin{enumerate}
		\item Put one node from the cycle of length $n^{\alpha}$ into the seed set.\label{part:one}
		\item Put one node from each of the visited cycles of length $n^{\beta}$ into the seed set. \label{part:2}
		\item Let the total number of the vertices selected above be $k'$. Our algorithm can choose $k-k'$ randomly selected nodes that it has not visited in the seed set. Since it has no information about those vertices, it is indifferent between all of them. \label{part:3}
	\end{enumerate}
	
	The coverage we obtain for Item~\eqref{part:one} is $n^{\alpha}$ nodes. Since the expected number of cycles of length $n^{\beta}$ that are visited by our algorithm is bounded by $2c \cdot n^{1-2\alpha + \beta + \gamma}$, the total coverage we obtain from nodes of Item~\eqref{part:2} is in expectation bounded by $2c \cdot n^{1-2\alpha + 2\beta + \gamma}$. For Item~\eqref{part:3}, we can bound $k-k'$ by $k$. Also, the expected number of nodes visited during our algorithm would be bounded by $n^{\alpha} + 2c \cdot n^{1-2\alpha + 2\beta + \gamma} + 2c\cdot n^{2-2\alpha}$ and therefore the expected number of unvisited nodes is at least $n - n^{\alpha} - 2c \cdot n^{1-2\alpha + 2\beta + \gamma} - 2c\cdot n^{2-2\alpha}$.
	
	The probability of selecting a node belonging to an unvisited cycle of length $n^{\beta}$ during our random selection is at most $\frac{n^{\beta+\gamma}}{n - o(n)} = O(n^{\beta+\gamma-1})$. Since we make at most $k = O( n^{\gamma})$ random choices, the expected number of nodes chosen from unvisited $n^{\beta}$-cycles is bounded by $O(n^{\gamma} \cdot n^{\beta+\gamma-1}) = O(n^{2\gamma+\beta-1})$. Each such node provides a coverage of $n^{\beta}$. Thus, the expected coverage from these random choices is $O(n^{2\gamma+2\beta-1})$.
	
	The total expected coverage of our algorithm is therefore bounded by $O(n^{\alpha} + c \cdot n^{1-2\alpha+2\beta+\gamma} + n^{2\gamma+2\beta-1})$. Substituting $\alpha = 1-\gamma/2$ and $\beta = 1-5\gamma/4$, we can evaluate each term independently:
	\begin{itemize}
		\item $n^{\alpha} = n^{1-\gamma/2}$
		\item $c \cdot n^{1 - 2(1-\gamma/2) + 2(1-5\gamma/4) + \gamma} = c \cdot n^{1 - 2 + \gamma + 2 - 5\gamma/2 + \gamma} = c \cdot n^{1-\gamma/2}$
		\item $n^{2\gamma + 2(1-5\gamma/4) - 1} = n^{2\gamma + 2 - 5\gamma/2 - 1} = n^{1-\gamma/2}$
	\end{itemize}
	Summing these, the algorithm's total coverage is strictly bounded by $O(c \cdot n^{1-\gamma/2})$.
	
	Meanwhile, the optimal coverage is $n^{\alpha} + n^{\beta+\gamma} = n^{1-\gamma/2} + n^{1-5\gamma/4+\gamma} = n^{1-\gamma/2} + n^{1-\gamma/4}$. Since $\gamma > 0$, we know $1-\gamma/4 > 1-\gamma/2$, meaning the optimal coverage is dominated by $\Omega(n^{1-\gamma/4})$.
	
	Taking the ratio of the algorithm's coverage to the optimal coverage yields an expected approximation factor of:
	$$ \frac{O(c \cdot n^{1-\gamma/2})}{\Omega(n^{1-\gamma/4})} = O(c \cdot n^{-\gamma/4}) $$
	
	Since $k = \Theta(n^{\gamma})$, it follows that $O(n^{-\gamma/4}) = O((n^\gamma)^{-1/4}) = O(k^{-1/4})$. Therefore, the expected approximation factor of the algorithm cannot be better than $O(c \cdot k^{-1/4})$, which concludes the proof.
\end{proof}

%% file: src/strong.tex
\section{Overcoming the Impossibility Result with Stronger Queries}\label{sec:strong}
As established in Section~\ref{sec:lowerbound}, any algorithm restricted to standard oracle queries and an $\tilde O(n)$ query cost budget is fundamentally limited to an $\tilde O(1/\sqrt{k})$ approximation guarantee. To bypass this barrier, we consider a stronger type of oracle access. Let $\mathcal{Q}_{s}$ be an oracle that takes a subset of elements $A \subseteq \mathcal{N}$ as input and returns a random set $S \in \mathcal{M}$ sampled uniformly and independently from the collection of sets that do not contain any element from $A$. The cost of this query remains equal to the cardinality of the returned set $S$. 

Before detailing the algorithm for this new setting, we establish a generalized, high-probability version of Lemma~\ref{lem:sample_size}. In our previous analysis, the sampling threshold was bounded by the optimal $k$-element coverage, $\mathsf{opt}$. Let $\mathsf{opt}_i$ denote the optimal coverage ratio when using $i$ elements. Specifically, $\mathsf{opt}_1 = \max_{e \in [n]} F(\{e\})$ represents the maximum single-element fractional coverage. We state Lemma~\ref{lem:generalized_sample_size} below and defer its proof to the end of this section.

\begin{lemma}\label{lem:generalized_sample_size}
	For some $\epsilon < 1/2$, let $\tau = \frac{60n \log n}{\epsilon^2}$ be the cumulative cost threshold. If we sample a collection $\mathcal{M}'$ of sets from $\mathcal{M}$ until the total cost reaches $\tau$, then with probability at least $1 - n^{-10}$, the number of sampled collections $|\mathcal{M}'| = m'$ satisfies $m' \ge \frac{30 \log n}{\epsilon^2 \mathsf{opt}_1}$.
\end{lemma}

Observe that Lemma~\ref{lem:generalized_sample_size} serves as a high-probability analogue of Lemma~\ref{lem:sample_size}, obtained by replacing the general optimum $\mathsf{opt}$ with the single-element optimum $\mathsf{opt}_1$ (effectively setting $k=1$). We now restate Lemma~\ref{lem:concentration} by parameterizing it with an integer $i \in [k]$. Rather than bounding the threshold by the $k$-element optimum coverage $\mathsf{opt}$, we generalize the statement in terms of $\mathsf{opt}_i$.

\begin{lemma} [restatement of Lemma~\ref{lem:concentration}]\label{lem:generalized_concentration}
	Let $i$ be an integer such that $1 \le i \le k$, and let $\mathcal{M}'$ be a collection of $m' \geq \frac{8 i \log n}{\epsilon^2 \mathsf{opt}_i}$ sets sampled uniformly and independently from $\mathcal{M}$. Let $F(S)$ denote the true fraction of sets in $\mathcal{M}$ covered by a fixed subset $S \subseteq \mathcal{N}$ such that $F(S) \le \mathsf{opt}_i$, and let $F'(S)$ denote the empirical fraction of sets in $\mathcal{M}'$ covered by $S$. For any $0 < \epsilon < 1$, the probability that the empirical coverage deviates from the true coverage by more than $\epsilon \cdot \mathsf{opt}_i$ is bounded by:
	$$\mathsf{Pr}\big[|F'(S) - F(S)| > \epsilon \cdot \mathsf{opt}_i\big] \le 2 n^{-3 i}.$$
\end{lemma}

The intuition behind our approach is as follows: Lemma~\ref{lem:generalized_sample_size} can be applied to Lemma~\ref{lem:generalized_concentration} for the specific case of $i=1$ to provide the error bounds necessary for proving the approximation factor of our algorithm. Furthermore, as long as the ratio $\mathsf{opt}_i / \mathsf{opt}_1$ remains proportional to $i$ (for instance, $\mathsf{opt}_i / \mathsf{opt}_1 \ge i/2$), one can argue that the sample set generated using the threshold from Lemma~\ref{lem:concentration} remains sufficient. Specifically, if that threshold is adjusted by a suitable constant factor, it can still be used to effectively approximate $\mathsf{opt}_i$. Conversely, this relationship no longer holds if the ratio falls below a certain threshold. In such cases, the existing sample set becomes insufficient to maintain the required error bounds, necessitating a new sampling step to accurately capture the desired properties.

This suggests a natural strategy: we run a greedy procedure on a fresh sample until the contribution of the elements to be added to our solution drops below $\mathsf{opt}_1/2$, at which point we discard the sample and \emph{reset} the process by setting $A = \mathcal{C}$ (setting $A$ equal to the current solution). In the next phase, we do the sampling for the residual problem by using the updated oracle $\mathcal{Q}_s(A)$. Thus, the algorithm proceeds in phases, where each phase consists of (i) sampling a representative collection $\mathcal{M}'$ for the current residual problem, and (ii) executing a greedy augmentation until the contributions of the remaining elements drop below $1/2$ of the contribution of the element that was added in the beginning of the phase. The algorithm terminates when either we construct a solution of size $k$ or the number of phases exceeds $\Theta(\log^5 n \log (1/\epsilon))$ at which point we terminate the algorithm and report our current solution.

Algorithm~\ref{algorithm:linear} implements our strategy. For the purposes of our discussion, we define $\mathsf{opt}^j_i$ to be the optimal solution for sets of size $i$ in the residual problem in the beginning of round $j$. In each phase, the algorithm samples a collection of sets $\mathcal{M}'$ from the stronger oracle $\mathcal{Q}_s$ until a cumulative cost threshold $(1000 \log^2 n)\tau$ is met, ensuring a statistically representative sub-collection for the current state. The additional multiplicative factor $1000 \log^2 n$ to the sample size is purposefully added to make up for the following requirements:
\begin{itemize}
	\item Although the error bound $\epsilon \cdot \mathsf{opt}^j_i$ provided in Lemma~\ref{lem:generalized_concentration} is small enough to prove the approximation factor of our algorithm, we need a slightly modified bound of $\frac{i  \cdot \mathsf{opt}^j_1}{20\log n}$ for the sake of runtime analysis. Notice that $\mathsf{opt}^j_i \geq \mathsf{opt}^j_1 \cdot \frac{i}{3}$ holds for any $i$ which is not bigger than the number of elements we add to our solution in phase $j$.
	\item To present a clean proof, we would like the error bound of Lemma~\ref{lem:generalized_concentration} to be $n^{-10i}$ instead of $2 n^{-3 i}$.
\end{itemize}
Thus, based on Lemma~\ref{lem:generalized_concentration}, in every phase $j$, for any subset $C$ of elements whose size is bounded by twice the number of elements we add to our solution in round $j$, we know that $|F'(C) - F(C)|$ is bounded by $\frac{|C|  \cdot \mathsf{opt}^j_1}{20\log n}$ with probability at least $1-n^{-10}$.

\begin{corollary}[of Lemma~\ref{lem:generalized_concentration}]\label{lemma:cor}
	For some $1 \leq j$, let $x$ be the number of elements added to the solution in phase $j$ of the algorithm. Then, with probability at least $1-n^{-10}$, for any subset $C$ such that $|C| \leq 2x$ we have:
	$$|F(C) - F'(C)| \leq |C|  \cdot \mathsf{opt}^j_1 \cdot \min\{\epsilon/2,\frac{1}{20\log n} \}$$
	where $F$ refers to the coverage function of the \textbf{residual} problem and $F'$ refers to the empirical estimation used in that phase of the algorithm.
\end{corollary}

\begin{corollary}[of Lemma~\ref{lem:generalized_concentration}]\label{lemma:cor2}
	Let $O$ be a set of elements that maximize the coverage for exactly $k$ elements in round 1. In any round $1 \leq j$ with probability at least $1-n^{-10}$ we have
	\begin{itemize}
	\item either the elements that have been added to $\mathcal{C}$ so far make a solution with approximation factor $9/10$.
	\item or $|F(O) - F'(O)| \leq \epsilon F(O)$ where $F$ and $F'$ are the coverage functions of the residual problems.
	\end{itemize}
\end{corollary}

For simplicity, we ignore the failure probability of Corollary~\ref{lemma:cor} in the rest of this section and only include it in our final theorem.

 Within these phases, Algorithm~\ref{algorithm:linear} greedily selects elements that maximize marginal gain based on the empirical coverage in $\mathcal{M}'$; however, unlike a standard greedy approach, it monitors the quality of the sample by comparing the current marginal gain against the initial gain of that phase. If the ratio of the current gain to the first gain drops below a $1/2$ threshold, the algorithm determines that the current sample no longer provides a reliable approximation for the remaining elements and thus we move on to the next phase. This mechanism ensures that the algorithm only makes decisions when the empirical data is sufficiently robust to maintain the desired high-probability error bounds throughout the selection process.

\begin{algorithm}[H]
	\caption{Adaptive Greedy Coverage}\label{algorithm:linear}
	\label{alg:adaptive_greedy}
	\begin{algorithmic}[1]
		\STATE \textbf{Input:} Adaptive Oracle $\mathcal{Q}_s$, ground set $\mathcal{N}$, budget $k$, error $\epsilon$
		\STATE \textbf{Initialize:} $\mathcal{C} \gets \emptyset$, $A \gets \emptyset$, $\tau \gets \frac{60n \log n}{\epsilon^2}$
		\STATE $\mathsf{iterations} \leftarrow 0$
		\WHILE{$|\mathcal{C}| < k$ \AND \textsf{iterations} $< c^* \log^5 n \log (1/\epsilon)$}
		\STATE $\mathsf{iterations} \leftarrow \mathsf{iterations} + 1$
		\STATE $\mathcal{M}' \gets \emptyset, q \gets 0$
		\WHILE{$q < (1000 \log^2 n)\tau$}
		\STATE $S \gets \mathcal{Q}_s(A)$; $\mathcal{M}' \gets \mathcal{M}' \cup \{S\}$; $q \gets q + |S|$
		\ENDWHILE
		\STATE $\mathsf{first\_gain} \gets \mathsf{null}$
		\WHILE{$|\mathcal{C}| < k$}
		\STATE $u \gets \arg\max_{e \in \mathcal{N}} |\{S \in \mathcal{M}' : e \in S \text{ and } S \cap \mathcal{C} = \emptyset \}|$
		\STATE $\mathsf{gain} \gets |\{S \in \mathcal{M}' : u \in S \text{ and } S \cap \mathcal{C} = \emptyset \}|$
		\IF{$\mathsf{first\_gain} == \mathsf{null}$}
		\STATE $\mathsf{first\_gain} \gets \mathsf{gain}$
		\ENDIF
		\IF{$\mathsf{gain} < \frac{\mathsf{first\_gain}}{2}$}
		\STATE $A \gets \mathcal{C}$
		\STATE \textbf{break} 
		\ENDIF
		\STATE $\mathcal{C} \gets \mathcal{C} \cup \{u\}$
		\ENDWHILE
		\ENDWHILE
		\RETURN $\mathcal{C}$
	\end{algorithmic}
\end{algorithm}

Before presenting our main result, we provide a brief justification for the termination condition of Algorithm~\ref{algorithm:linear}. This stopping criterion is essential to ensure that both the total sample cost and the overall runtime of the algorithm are bounded by $\tilde{O}_{\epsilon}(n)$. In Lemma~\ref{lemma:kuchik}, we show that the marginal gains of the remaining elements diminish exponentially as the number of iterations increases. Specifically, it can be implied from Lemma~\ref{lemma:kuchik} that after $O(\log^5 n \log 1/ \epsilon)$ iterations, the remaining potential gains are sufficiently small to be ignored without compromising the approximation factor. For clarity, the formal proof of Lemma~\ref{lemma:kuchik} is deferred to the end of this section.

\begin{lemma}\label{lemma:kuchik}
	There exists a constant $c^*$ such that after $c^* \log^5 n \log (1/\epsilon)$ phases of Algorithm~\ref{algorithm:linear} the marginal contributions of the remaining elements fall below $\epsilon /n$.
\end{lemma}

We also state an algebraic statement in Lemma~\ref{lem:lower_bound_sum} which will be used in the proof of Theorem~\ref{theorem:linear}. We will give a proof for Lemma~\ref{lem:lower_bound_sum} in Section~\ref{sec:sag}.

\begin{lemma}\label{lem:lower_bound_sum}
	Let $l$ be a positive integer, $u$ be a positive real number, $\alpha \in (0, 1]$, and $\beta > 1$. Suppose that $d_1, d_2, \dots, d_l$ is a sequence of positive real numbers such that for any $i \in \{1, 2, \dots, l\}$, the following inequality holds:
	$$
	d_i \ge \left( u - \beta \sum_{j=1}^{i-1} d_j \right) \frac{\alpha}{l}.
	$$
	Then, the total sum of the sequence satisfies the lower bound
	$$
	\sum_{i=1}^l d_i \ge u \left( 1 - \frac{1}{e} \right)  \min\{\alpha, 1/\beta\}.
	$$
\end{lemma}
We are now ready to prove Theorem~\ref{theorem:linear}.

\begin{theorem}\label{theorem:linear}
	If query oracle $\mathcal{Q}_s$ is available, Algorithm~\ref{algorithm:linear} runs in time $O(\frac{n \log^9 n  \log (1/\epsilon)}{\epsilon^2})$, has query cost $O(\frac{n \log^8 n  \log (1/\epsilon)}{\epsilon^2})$, and achieves a $(1-1/e-4\epsilon)$-approximate solution for the maximum $k$-cover problem with probability at least $1-n^{-8}$.
\end{theorem}
\begin{proof}
	
	For the runtime and query cost complexity of the algorithm, notice that the number of rounds is bounded by $O(\log^5 n \log (1/\epsilon))$. In each round, the sample size is bounded by $O(n \log^3 n / \epsilon ^2)$ and thus the overall query complexity would be bounded $O(\frac{n \log^8 n  \log (1/\epsilon)}{\epsilon^2})$. The runtime would be the same except that we use a heap data structure to find the element with the highest contribution in each turn that adds a multiplicative factor of $O(\log n)$. In the rest of the proof, we show that the approximation ratio of the algorithm is at least $(1-1/e - 4\epsilon)$. 
	
	We assume for simplicity here that the algorithm has no failure and discuss the failure probability at the end of the proof. Also, we assume for simplicity here that Algorithm~\ref{algorithm:linear} iteratively puts $k$ elements into the solution and the early termination does not happen. We discuss at the end of the proof that the early termination does not deteriorate the quality of our solution substantially.
	
	To this end, let $d_1, d_2, \ldots, d_k$ be the marginal contributions of the elements that Algorithm~\ref{algorithm:linear} iteratively puts into the solution to the function $F(\mathcal{C})$. Also, we define $d'_i$ as the same contribution for the sampled problem of the corresponding round. That is, let $F'$ be the sampled $k$-cover problem of the corresponding round. Instead of function $F$, we use $F'$ to determine the value of $d'_i$. We also define $u$ as the highest coverage of any subset of size $k$ which is the coverage ratio of the optimal solution. Notice that for each value $d'_i$, we consider the marginal gain with respect to the original problem (and not only the problem with the avoiding set).

	Now, we wish to apply Lemma~\ref{lem:lower_bound_sum} to give a lower bound on the quality of our algorithm. Fix some $1 \leq r \leq k$ and assume that the $r$'th element of the solution is added in round $j$. Also, assume that $q$ elements are added to the solution prior to round $j$. Define $v$ as the highest utility a set of size $k$ can obtain in the beginning of round $j$ in the residual problem. It follows that $v \geq u - \sum_{i=1}^q  d_i$. Also, define $v'$ as the coverage ratio of the same subset in the sampled problem. Based on Corollary~\ref{lemma:cor2} we have $v' \geq v (1-\epsilon) \geq (u - \sum_{i=1}^q  d_i)(1-\epsilon)$. Moreover, it follows from Corollary~\ref{lemma:cor} that $\sum_{i=1}^q  d_i \leq (1+\epsilon) \sum_{i=1}^q  d'_i$ and therefore $v' \geq (u - (1+\epsilon)\sum_{i=1}^q  d'_i)(1-\epsilon)$. This implies that at the time of adding the $r$'th element to the solution, there is a subset of size $k$ whose empirical contribution with respect to the sampled problem is at least 
\begin{align*}
	v' - \sum_{i=q+1}^{r-1} d'_i &\geq \left( u - (1+\epsilon)\sum_{i=1}^q d'_i \right)(1-\epsilon) - \sum_{i=q+1}^{r-1} d'_i \\
	&= u(1-\epsilon) - (1-\epsilon^2)\sum_{i=1}^q d'_i - \sum_{i=q+1}^{r-1} d'_i \\
	&\geq u(1-\epsilon) - \sum_{i=1}^q d'_i - \sum_{i=q+1}^{r-1} d'_i \\
	&= u(1-\epsilon) - \sum_{i=1}^{r-1} d'_i \\
	&\geq u(1-\epsilon) - (1+2\epsilon)(1-\epsilon)\sum_{i=1}^{r-1} d'_i \\ 
	&= \left( u - (1+2\epsilon)\sum_{i=1}^{r-1} d'_i \right)(1-\epsilon)
\end{align*}
Here, the last inequality in the display above uses $(1+2\epsilon)(1-\epsilon) = 1+\epsilon-2\epsilon^2 \geq 1$, which holds since $\epsilon < 1/2$ (as in Lemma~\ref{lem:generalized_sample_size}).

	And therefore since each time we add the element with the highest marginal contribution, we have $d'_i \geq \left( u - (1+2\epsilon)\sum_{i=1}^{r-1} d'_i \right)\frac{(1-\epsilon)}{k}$. Applying Lemma~\ref{lem:lower_bound_sum} to $d'$ with $\beta = 1+2\epsilon$ and $\alpha = 1-\epsilon$ implies $\sum_{i=1}^k d'_i \ge u \left( 1 - \frac{1}{e} \right)  (1-2\epsilon)$. This in addition to Corollary~\ref{lemma:cor} implies $\sum_{i=1}^k d_i \ge u \left( 1 - \frac{1}{e} \right)  (1-2\epsilon)(1-\epsilon) \geq  u \left( 1 - \frac{1}{e} \right)  (1-3\epsilon)$ which is desired.
	
	In the case of early termination, we know that the future contribution of each neglected element is bounded by $\epsilon / n$ and therefore our solution may be deteriorated by an additive error of at most $\frac{k \epsilon}{n}$ which introduces a multiplicative error of at most $\epsilon$ to our approximation guarantee since any solution of size $k$ has a coverage of at least $k/n$. Therefore the approximation factor of our algorithm is bounded by $(1-1/e - 4\epsilon).$ Also, since each failure probability is bounded by $n^{-10}$ the number of such possible scenarios is bounded by $n^2$, the total failure probability of the algorithm is bounded by $n^{-8}$.
\end{proof}

\subsection{Proof of Lemma~\ref{lem:generalized_sample_size}}
\begin{proof}[of lemma~\ref{lem:generalized_sample_size}]
	Let $S_1, S_2, \dots$ be the sequence of sets sampled from $\mathcal{M}$. Let $x_i = |S_i|/n$ be the normalized cost of the $i$-th query, where $x_i \in [0, 1]$. Let $m_l = \frac{30 \log n}{\epsilon^2 \mathsf{opt}_1}$. We define the total normalized cost after $m_{l}$ samples as $x = \sum_{i=1}^{m_{l}} x_i$. 
	
	To bound the deviations of $x$, we employ the following standard form of the Chernoff bound~\cite{mitzenmacher2017probability}: Let $x_1, \dots, x_{\alpha}$ be independent random variables taking values in $[0, 1]$ and $x = \sum_{i=1}^{\alpha} x_i$. For any $\mu_{u}$ such that $\mathbb{E}[x] \le \mu_{u}$ and any $\delta > 0$:
	\[ \Pr[x > (1+\delta)\mu_{u}] \le \exp\left(-\frac{\delta^2 \mu_{u}}{2 + \delta}\right) \]
	
	To prove the lemma, we consider the event that the cost threshold is exceeded before $m_{l}$ sets are sampled. This occurs only if $x > \tau/n = \frac{60 \log n}{\epsilon^2}$. From the double-counting argument of Lemma~\ref{lem:sample_size}, we have $\mathbb{E}[x_i] \le \mathsf{opt}_1$, and therefore $\mathbb{E}[x] \le m_{l} \cdot \mathsf{opt}_1 = \frac{30 \log n}{\epsilon^2}$ holds. Let $\mu_{u} = \frac{30 \log n}{\epsilon^2}$ be our upper bound on the expectation. Applying the Chernoff bound with $\delta = 1$:
	\[ \Pr[x > 2\mu_{u}] \le \exp\left(-\frac{\mu_{u}}{3}\right) = \exp\left(-\frac{10 \log n}{\epsilon^2}\right) \le n^{-10} \]
	Thus, with probability at least $1 - n^{-10}$, the cost of the first $m_{l}$ sets does not exceed $\tau/n$, which implies $|\mathcal{M}'| \ge m_{l}$.
\end{proof}

\subsection{Proof of Lemma~\ref{lemma:kuchik}}
For the sake of analysis, let us assume that we run the algorithm until all $n$ elements are  added to the solution (this is just a thought experiment and we do not actually run the algorithm to generate values $g_i$). For each $e \in \mathcal{N}$, define $g_e$ to be the marginal contribution of element $e$ to the coverage of our solution at the time of adding it to our solution. Based on this definition, the following property follows: During the lifetime of our algorithm, for each set of elements $C$ that are not yet added to the solution, their marginal contribution to the existing solution is at least $\sum_{e \in C} g_e$. 

\begin{lemma} \label{lem:marginal_decay}
	Let $C^1$ and $C^2$ be the solution of Algorithm~\ref{algorithm:linear} before and after a phase of the algorithm. For any $0 \leq x \leq 1$, define 
	$$\Delta^1_x = \{e \in \mathcal{N} \setminus C^1 | g_e \geq x \}$$ and $$\Delta^2_x = \{e \in \mathcal{N} \setminus C^2 | g_e \geq x \}.$$
	Then for any $\kappa$ such that $\frac{\max_{e \in \mathcal{N} \setminus C^1} g_e}{2 (1-1/\log n)} \leq \kappa $ holds we have 
	$$|\Delta^1_{\kappa (1-1/\log n)} \setminus \Delta^2_{\kappa (1-1/\log n)}| \geq \lceil \frac{|\Delta^1_{\kappa}|}{3 \log n} \rceil .$$
\end{lemma}
\begin{proof}
	Let $\alpha = |\Delta^1_{\kappa}|$ denote the number of elements not yet in $C^1$ whose marginal contribution to the final solution is at least $\kappa$.  If at least half of the elements of $\Delta^1_{\kappa}$ are added to the solution in this phase, then the condition of the lemma holds trivially. Moreover, the same holds for $\Delta^1_{\kappa (1-1/\log n)}$. That is, if at least $\lceil \frac{\alpha}{3 \log n} \rceil$ elements of $\Delta^1_{\kappa (1-1/\log n)}$ are added to the solution in this phase the condition of the lemma holds trivially. Therefore, we may assume without loss of generality that at least $\frac{\alpha}{2}$ elements from $\Delta^1_{\kappa}$ are not added to the solution in this phase and moreover, at most $\alpha - \lceil \frac{\alpha}{3 \log n} \rceil$ elements of $\Delta^1_{\kappa (1-1/\log n)}$ are added to the solution in this phase. We now consider two mutually exclusive cases based on the total number of elements added to the solution during the phase:
	
\textbf{Case 1:} The number of elements added to the solution in the phase is at least $\alpha$. For simplicity and without loss of generality we assume $\alpha$ is even.
Let $C$ be the set of the first $\alpha$ elements that are added to the solution in this phase. For the condition of the lemma to fail, strictly fewer than $\lceil \frac{\alpha}{3 \log n} \rceil$ of the elements added have a real contribution of at least $\kappa(1-\frac{1}{\log n})$. We define $C^g$ as a set of exactly $\frac{\alpha}{2}$ elements of $\Delta^1_{\kappa}$ that were not added to the solution. Because each element in $C^g$ maintains a real marginal contribution of at least $\kappa$ even after the phase, their joint real contribution evaluated over $C$ is bounded from below by:
$$F(C \cup C^g) \ge F(C) + \frac{\alpha}{2}\kappa.$$
Applying Corollary~\ref{lemma:cor}, the empirical error for any set $S$ of size up to $2\alpha$ is at most $\frac{|S| \cdot \mathsf{opt}^j_1}{20\log n} \le \frac{|S|\kappa}{10\log n}$. \claudeedit{Here, we use $\mathsf{opt}^j_1 \le 2\kappa$, which holds since the first element $e^f$ added in this phase satisfies $g_{e^f} \ge (1-\frac{1}{10\log n})\mathsf{opt}^j_1$ by Corollary~\ref{lemma:cor}, and hence $\mathsf{opt}^j_1 \le \frac{\max_{e \in \mathcal{N} \setminus C^1} g_e}{1-1/(10\log n)} \le \frac{2\kappa(1-1/\log n)}{1-1/(10\log n)} < 2\kappa$.}  Therefore, the empirical contribution of $C^g$ over $C$ in the sampled set $\mathcal{M}'$ satisfies:
$$F'(C \cup C^g) - F'(C) \ge \left(F(C \cup C^g) - \frac{|C \cup C^g|\kappa}{10\log n}\right) - \left(F(C) + \frac{|C|\kappa}{10\log n}\right).$$
Substituting $|C| = \alpha$ and $|C \cup C^g| = \frac{3\alpha}{2}$, we obtain:
$$F'(C \cup C^g) - F'(C) \ge \frac{\alpha}{2}\kappa - \frac{3\alpha\kappa}{20\log n} - \frac{2\alpha\kappa}{20\log n} = \frac{\alpha}{2}\kappa\left(1 - \frac{1}{2\log n}\right).$$
This implies that the total empirical marginal contribution of the elements in $C^g$ over $C$ is at least $\frac{\alpha}{2}\kappa(1 - \frac{1}{2\log n})$. By the pigeonhole principle, there must exist at least one element $e^* \in C^g$ whose empirical marginal contribution to $C$ in $\mathcal{M}'$ is at least the average, which is $\kappa(1 - \frac{1}{2 \log n})$. By submodularity, at any step during the construction of $C$, the available empirical marginal contribution of $e^*$ to the partial solution was also at least $\kappa(1 - \frac{1}{2\log n})$. This implies that the marginal empirical contribution of any element added to $C$ during the first $\alpha$ steps is at least this amount which implies that their total empirical contribution is at least $\alpha \kappa(1 - \frac{1}{2\log n})$. This implies that their real contribution is at least $\alpha \kappa(1 - \frac{1}{2\log n} - \frac{1}{10\log n}) = \alpha \kappa(1  - \frac{6}{10\log n})$. This is in contradiction with the fact that apart from $\lceil \frac{\alpha}{3 \log n} \rceil -1 $ of such element whose real contribution could be up to $2\kappa$, the rest have a real contribution less than $\kappa (1-1/\log n)$ and thus their total real contribution would be bounded by $\alpha \kappa(1  - \frac{2}{3\log n})$.
 
\textbf{Case 2}: The number of elements added to the solution in the phase is $\beta < \alpha$. Again, for simplicity and without loss of generality we assume $\beta$ is even. Let $C$ denote the set of $\beta$ elements added during this phase. As we discussed above, there are at least $\alpha/2$ elements whose real marginal contribution remains at least $\kappa$ after this phase ends. Since $\beta < \alpha$, we have $\alpha/2 > \beta/2$. We can therefore select a subset $C^g$ of $\beta/2$ such good elements that were not added to the solution. 

The size of the union $C \cup C^g$ is $\frac{3}{2}\beta$, which is less than $2\beta$, satisfying the condition of Corollary~\ref{lemma:cor}. Because every element in $C^g$ has a real marginal contribution of at least $\kappa$ even after the phase ends, we have:
$$F(C \cup C^g) \geq F(C) + \frac{\beta}{2}\kappa$$
Applying Corollary~\ref{lemma:cor}, the total empirical error is at most $\frac{|C \cup C^g|\kappa}{10\log n} = \frac{3\beta\kappa}{20\log n}$. The empirical marginal contribution of the set $C^g$ over $C$ in $\mathcal{M}'$ is then:
$$F'(C \cup C^g) - F'(C) \geq \left(F(C \cup C^g) - \frac{3\beta\kappa}{20\log n}\right) - \left(F(C) + \frac{2\beta\kappa}{20\log n}\right) = \frac{\beta}{2}\kappa \left(1 - \frac{1}{2 \log n}\right)$$
By an averaging argument, there must exist at least one element $e^* \in C^g$ whose empirical marginal contribution to $C$ is at least $F'(C \cup \{e^*\}) - F'(C) \geq \kappa(1 - \frac{1}{2\log n})$. Let $e^f$ be the first element that was added to the solution in this phase. It follows from Corollary~\ref{lemma:cor} and the maximality of the empirical coverage of $e^f$ that $F'(\{e^f\}) \leq \mathsf{opt}^j_1 (1+1/(20 \log n))$. Since the phase only terminates when all available empirical contributions fall below the threshold $$\frac{F'(\{e^f\})}{2} \leq \frac{\mathsf{opt}^j_1 (1+1/(20 \log n))}{2} \leq \frac{2\kappa(1-1/\log n) (1+1/(20 \log n))}{2} < \kappa (1-\frac{1}{2\log n})$$ the existence of $e^*$ with an empirical contribution of at least $\kappa (1-\frac{1}{2\log n})$ contradicts the termination of the phase.
\end{proof}

Now, we are ready to prove Lemma~\ref{lemma:kuchik}.

\begin{proof}[of Lemma~\ref{lemma:kuchik}]
	The proof follows from Lemma~\ref{lem:marginal_decay}. Consider the first phase of the algorithm and let $h$ be the highest contribution of any element. Consider $O(\log n)$ values $\Delta_{h(1-1/\log n)},\Delta_{h(1-1/\log n)^2}, \ldots, $ where $\Delta_x$ is the number of elements whose $g_i$ is at least $x$. As we add elements to the solution, we update $\Delta_x$ values by only considering the elements that are not added to the solution. Since we have $O(\log n)$ elements that are increasing, and their maximum is $n$, in every phase, there is one $x$ such that $\Delta_x$ is at most constant times more than $\Delta_{x/(1-1/\log n)}$ and thus after that phase, according to Lemma~\ref{lem:marginal_decay} the value $\Delta_x$ is multiplied by a fraction of $1-\Omega(1/\log n)$. \claudeedit{To see this, note that the quantity $\sum_x \ln(1+\Delta_x)$, taken over the above $O(\log n)$ values, is at most $O(\log^2 n)$, it never increases, and it decreases by $\Omega(1/\log n)$ in every phase. Moreover, since $\Delta_x \le \Delta_{x'}$ whenever $x > x'$, the first of these values that reduces to $0$ is the top one.} Thus, after $O(\log^3 n)$ steps, one of the above values reduces to 0 meaning that there remains no element that is not in the solution whose value $g_i$ is at least $h(1-1/\log n)$. \claudeedit{We then repeat the same argument with the new highest contribution in place of $h$ (which is nonzero by definition, so the top value of the new window is nonzero). Therefore, the highest contribution decreases by a factor of $1-1/\log n$ every $O(\log^3 n)$ phases, and it takes $O(\log n \log(n/\epsilon))$ such decreases to go from at most $1$ to below $\epsilon/n$.} This means that after $O(\log^5 n \log (1/\epsilon))$ phases, the contribution of each of the remaining elements is bounded by $\epsilon / n$.
\end{proof}

\subsection{Proof of Lemma~\ref{lem:lower_bound_sum}}\label{sec:sag}
\begin{proof}[of Lemma~\ref{lem:lower_bound_sum}]
	If $\alpha\beta > l$, the claim follows immediately: in this case $\alpha/l > 1/\beta$, and the inequality for $i=1$ gives $\sum_{i=1}^l d_i \ge d_1 \ge u\alpha/l > u/\beta \ge u\left(1-\frac{1}{e}\right)\min\{\alpha, 1/\beta\}$. Therefore, we assume in the rest of the proof that $\alpha\beta \le l$, i.e., $1-\alpha\beta/l \ge 0$.
	We track the remaining residual value with respect to $u$ after subtracting the scaled prefix sums. For each index $i \in \{1, 2, \dots, l\}$, expanding the summation yields the algebraic identity
	$$
	u - \beta \sum_{j=1}^i d_j = u - \beta \sum_{j=1}^{i-1} d_j - \beta d_i.
	$$
	Substituting the lower bound for $d_i$ into this expression establishes the upper bound
	$$
	u - \beta \sum_{j=1}^i d_j \le u - \beta \sum_{j=1}^{i-1} d_j - \beta \left( u - \beta \sum_{j=1}^{i-1} d_j \right) \frac{\alpha}{l}.
	$$
	Factoring out the common residual term simplifies the recurrence relation to
	$$
	u - \beta \sum_{j=1}^i d_j \le \left( u - \beta \sum_{j=1}^{i-1} d_j \right) \left( 1 - \frac{\alpha\beta}{l} \right).
	$$
	Applying this inequality inductively from $i = l$ down to $i = 1$ (which is valid since $1-\alpha\beta/l \ge 0$) yields
	$$
	u - \beta \sum_{i=1}^l d_i \le u \left( 1 - \frac{\alpha\beta}{l} \right)^l.
	$$
	To decouple the bound from the sequence length $l$, we apply the standard analytical inequality $1 - x \le e^{-x}$~\cite{mitzenmacher2017probability}, which holds for all real $x$. Setting $x = \frac{\alpha\beta}{l}$ produces
	$$
	u \left( 1 - \frac{\alpha\beta}{l} \right)^l \le u \left( e^{-\alpha\beta/l} \right)^l = u e^{-\alpha\beta}.
	$$
Combining these inequalities and isolating the total sum yields the intermediate bound
$$
\sum_{i=1}^l d_i \ge \frac{u}{\beta} \left( 1 - e^{-\alpha\beta} \right).
$$
We utilize the concavity of the function $f(x) = 1 - e^{-x}$. Because $f(x)$ is strictly concave, its graph lies above the secant line connecting $(0,0)$ and $(1, 1 - e^{-1})$ on the interval $[0, 1]$, giving $1 - e^{-x} \ge (1 - e^{-1})x$ for $0 \le x \le 1$. Furthermore, since $f(x)$ is strictly increasing, $1 - e^{-x} \ge 1 - e^{-1}$ for all $x \ge 1$. Combining these observations yields the global lower bound $1 - e^{-x} \ge (1 - e^{-1}) \min(x, 1)$ for all $x \ge 0$. Evaluating this expression at $x = \alpha\beta$ implies
$$
\sum_{i=1}^l d_i \ge \frac{u}{\beta} \left( 1 - e^{-\alpha\beta} \right) \ge \frac{u}{\beta} (1 - e^{-1}) \min(\alpha\beta, 1) = u (1 - e^{-1}) \min(\alpha, 1/\beta)
$$

\begin{figure}[htbp]
	\centering
	\begin{tikzpicture}
		\begin{axis}[
			axis lines = middle,
			xlabel = {$x$},
			ylabel = {},
			xmin = -0.1, xmax = 2.2,
			ymin = -0.1, ymax = 1.1,
			xtick = {0, 1, 2},
			xticklabels = {0, 1, 2},
			ytick = {0, 0.6321, 1},
			yticklabels = {0, $1 - e^{-1}$, 1},
			legend pos = south east,
			grid = none,
			width = 8cm, height = 6cm
			]
			\addplot [
			domain = 0:2.1, 
			samples = 100, 
			color = blue,
			thick
			] {1 - exp(-x)};
			\addlegendentry{$1 - e^{-x}$}
			
			\addplot [
			domain = 0:2.1, 
			samples = 100, 
			color = red,
			dashed,
			thick
			] {(1 - exp(-1)) * (x < 1 ? x : 1)};
			\addlegendentry{$(1 - e^{-1})\min(x, 1)$}
			
			\draw[gray, dashed] (axis cs:1,0) -- (axis cs:1,0.6321);
			\draw[gray, dashed] (axis cs:0,0.6321) -- (axis cs:2.1,0.6321);
			
			\node[circle, fill=black, inner sep=1.5pt] at (axis cs:1,0.6321) {};
		\end{axis}
	\end{tikzpicture}
	\caption{Comparison between the function $g(x) = 1 - e^{-x}$ and the piecewise linear lower bound valid for all $x \ge 0$.}
	\label{fig:concavity_bound}
\end{figure}
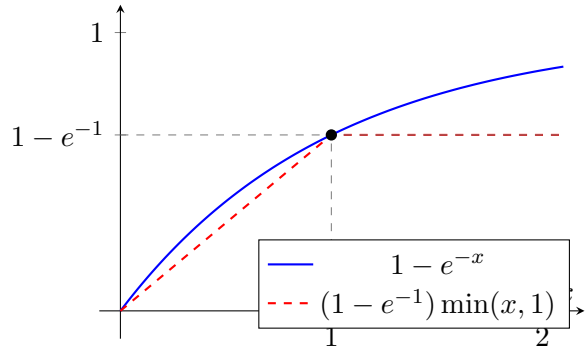

which completes the proof.
\end{proof}

%% file: src/query.tex
\section{Stronger Queries for Social Networks}\label{sec:practical_queries}
In Section~\ref{sec:strong}, we established that by employing a stronger type of query oracle, $\mathcal{Q}_s$, one can achieve a $(1-1/e-\epsilon)$ approximation guarantee for the maximum $k$-coverage problem in almost linear time and with almost linear query cost. This result stands in sharp contrast to the $O(1/\sqrt{k})$ bound derived in Section~\ref{sec:lowerbound}, which applies to any algorithm restricted to standard oracle queries. The gap between these two results stems directly from the enhanced information provided by $\mathcal{Q}_s$ which is an oracle that returns sets sampled specifically from the sub-collection of $\mathcal{M}$ that avoid a prohibited set of elements $A$.

However, to maintain the practical utility of our result, these stronger queries must be efficiently implementable within the standard social network diffusion models. Recall from Section~\ref{sec:reduction} that a standard query $\mathcal{Q}$ is answered by selecting a root node $v \in V$ uniformly at random and performing a graph traversal (e.g., BFS or DFS) on the reversed edges of a sampled realization. The resulting set of reachable nodes constitutes the reverse reachable set. To implement $\mathcal{Q}_s$, the oracle must return a reverse reachable set that does not contain any node from a given set $A$.

A naive approach to simulating $\mathcal{Q}_s$ using the standard oracle $\mathcal{Q}$ is \textit{rejection sampling}: one repeatedly draws reverse reachable sets using $\mathcal{Q}$ until a set is found that satisfies the condition $R \cap A = \emptyset$. While this is theoretically sound, it introduces a significant computational bottleneck. The cost of a single successful query to $\mathcal{Q}_s$ would become the sum of the costs of all discarded samples plus the cost of the last successful sample. 

To get intuition, consider a case where all collections in $\mathcal{M}$ are of the same size. In this case, the simple approach is actually desirable. For an avoiding set $A$, define $r_A$ to be the ratio of collections of $\mathcal{M}$ that do not contain any elements from $A$. It is easy to show that with high probability ($1-n^{-c}$ for an arbitrarily large constant $c$), the number of unsuccessful attempts to answer $\mathcal{Q}_s(A)$ is bounded by $O(\log n/r_A)$. Thus, as long as $r_A > \epsilon$, this only adds a multiplicative factor of $\tilde{O}_{\epsilon}(1)$ to the cost function, leaving the overall complexity of the algorithm almost intact. 

This is indeed not the case when $r_A$ drops below $\epsilon$. However, notice that in Algorithm~\ref{algorithm:linear}, when querying $\mathcal{Q}_s$, we always choose $A$ to be the current elements that are present in our solution ($A = \mathcal{C}$). Thus, if $r_A$ drops below $\epsilon$, this means that our current solution is already $(1-\epsilon)$-approximate and therefore we can terminate the algorithm already.

Given that in Section~\ref{sec:lowerbound} we showed that $\mathcal{Q}$ is not strong enough to obtain an approximation factor better than $\tilde O(1/\sqrt{k})$, we expect the simple rejection sampling approach to fail if collection sizes are not the same. Let us remind the reader of the hard instance construction: we consider a ground set partitioned into two blocks where the multiset $\mathcal{M}$ contains $s \approx \sqrt{k}$ identical copies of a large collection $L= \{1, 2, \dots, n/2\}$ and $k-s$ singleton collections $\{x\}$ for $x$ chosen from $\{n/2+1, n/2+2, \ldots, n\}$.

In the first phase, our algorithm adds element $1$ to the solution set (since it is included in the large collections) and we now want to implement the oracle $\mathcal{Q}_s$ with the prohibited set $A = \{1\}$. Notice that all the collections in $\mathcal{M}$ that are valid for this query (those that do not contain element $1$) are the singleton collections, each having a size of $1$. However, if we use the naive rejection strategy, we repeatedly draw from the standard oracle $\mathcal{Q}$. Since $\mathcal{Q}$ returns the large collection $L$ with probability $s/k \approx 1/\sqrt{k}$, and each such draw must be discarded at a cost of $|L| = n/2$, the expected cost to successfully sample one valid singleton collections via Algorithm~\ref{algorithm:linear} escalates to $O(n/\sqrt{k})$. This effectively adds a multiplicative factor of $O(n/\sqrt{k})$ to the query cost, making it impossible to obtain a desirable solution with query cost $\tilde O_{\epsilon}(n)$.

Recall that our goal is to apply this framework to the influence maximization problem, where reverse reachable sets are generated by performing a traversal on the reversed edges of a social network. While the implementation of the standard oracle $\mathcal{Q}$ is straightforward via a graph search (e.g., BFS or DFS) starting from a randomly selected root, implementing $\mathcal{Q}_s$ with the required cost limit is not directly possible. The primary obstacle is that the only way to determine if a traversal will eventually reach a prohibited node in $A$ is to actually perform the search, which may incur a significant cost before the condition $R \cap A = \emptyset$ can be evaluated.

To bridge this gap, we introduce a middle-ground oracle, $\mathcal{Q}_{\ell}$, which serves two purposes: it is efficiently implementable via a graph search with a guaranteed cost bound, and it can be used to implement the stronger oracle $\mathcal{Q}_s$ effectively. The oracle $\mathcal{Q}_{\ell}$ receives an integer parameter $l$ as input and, when invoked, returns one of the following:
\begin{itemize}
	\item A collection of size at most $l$ which is the reverse reachable set of a randomly chosen node. In this case, the incurred cost is equal to the size of the returned collection.
	\item A \texttt{null} value, indicating that the reverse reachable set of the randomly chosen node has a size strictly larger than $l$. In this case, the cost of the oracle call is exactly $l$.
\end{itemize}
This oracle is easily implementable by a standard graph traversal that includes a breaking point: the search is terminated as soon as the number of explored edges or nodes exceeds $l$.

Back to the hard instance construction in Section~\ref{sec:lowerbound}, one can observe that $\mathcal{Q}_{\ell}$ can resolve our problem easily. In that scenario, the bottleneck was the high cost of sampling the large set $L$ only to discard it because it contained the prohibited element. By invoking $\mathcal{Q}_{\ell}$ with a small threshold $l = 1$, the oracle returns \texttt{null} as soon as it realizes that the collection it is drawing has a size more than 1. This allows the algorithm to skip over the bad large sets efficiently and successfully sample the valid singleton collections (which are of size $1$) within each with a cost of $\tilde{O}_{\epsilon}(1)$. Thus, $\mathcal{Q}_{\ell}$ may provide a way to simulate the necessary properties of $\mathcal{Q}_s$ without the prohibitive computational overhead of full rejection sampling on large sets.

We show in the following that $\mathcal{Q}_{\ell}$ is not merely a specialized fix to bypass the specific difficulty captured by the lower bound in Section~\ref{sec:lowerbound}; rather, it is a robust tool for efficiently implementing $\mathcal{Q}_s$ in general. By carefully choosing the threshold $l$, we can use $\mathcal{Q}_{\ell}$ to simulate the behavior of the stronger oracle $\mathcal{Q}_s$ across diverse instances, ensuring that large, invalid sets never dominate the computational cost. This makes it a fundamental building block for achieving our almost linear-time approximation algorithm.

It is worth examining the intuition behind the lower bound presented in Section~\ref{sec:lowerbound}. One might observe that the primary reason oracle $\mathcal{Q}$ fails to provide a decent approximation for the aforementioned example is that the average size of the collections of $\mathcal{M}$ is significantly larger than the average size of the collections of $\mathcal{M}$ that do not contain element $1$. Otherwise, if the average sizes were comparable, one could apply an argument similar to the uniform-size case. Based on this intuition, we present Lemma~\ref{lem:threshold} that provides the technical threshold $l$ necessary to bridge the gap between $\mathcal{Q}$ and $\mathcal{Q}_s$ by identifying a scale where the expected sizes and the tail probabilities are balanced relative to $\epsilon$. For brevity, we defer the proof of Lemma~\ref{lem:threshold} to Section~\ref{sec:prooflemmal}.

\begin{lemma}
	\label{lem:threshold}
	Let $\mathcal{N} = \{1, 2, \ldots, n\}$ and let $\mathcal{M}$ be a multiset of non-empty subsets of $\mathcal{N}$. For any avoiding set $A \subseteq \mathcal{N}$, define $\mathcal{M}^* = \{X \in \mathcal{M} \mid X \cap A = \emptyset\}$. Assume that $|\mathcal{M}^*| > \epsilon |\mathcal{M}|$ for some constant $\epsilon > 0$. Let $\Delta = \lceil \log_2 n \rceil+2$. Then, there exists a threshold $l \in \{0, 1, 3, 7, \ldots, 2^{\Delta-1}-1\}$ such that the following two conditions hold simultaneously:
	\begin{enumerate}
		\item $|\{X \in \mathcal{M}^* \mid |X| > l\}| \leq \epsilon |\{X \in \mathcal{M} \mid |X| > l\}|$
		\item $\mathbb{E}_{X \sim \mathcal{M}}[\min\{|X|, l\}] \leq \frac{6 \Delta^2}{\epsilon^2} \cdot \mathbb{E}_{X \sim \mathcal{M}^* \text{ such that }|X| \leq l}[|X|]$
	\end{enumerate}
\end{lemma}

To provide further clarity for Lemma~\ref{lem:threshold}, we define a notion of similarity between multisets. Given a subset of elements $A \subseteq \mathcal{N}$, a multiset $\hat{\mathcal{M}}$ is said to be $(A, \epsilon)$-similar to $\mathcal{M}$ if the following conditions are satisfied:

\begin{itemize}
	\item $\hat{\mathcal{M}} \subseteq \mathcal{M}$
	\item $|\{X \in \mathcal{M} \setminus \hat{\mathcal{M}} \mid X \cap A = \emptyset\}| \leq \epsilon |\mathcal{M} \setminus \hat{\mathcal{M}}|$
\end{itemize}

Intuitively, this similarity condition implies that if we replace $\mathcal{M}$ with an $(A, \epsilon)$-similar multiset during the execution of our greedy algorithm (where $A$ represents the current set of chosen elements), the additive error in the solution is bounded by at most $2\epsilon \mathsf{opt}$ where $\mathsf{opt}$ is the coverage of the optimal set of size $k$ (Recall that $\mathsf{opt} \geq \frac{|\{X \in \mathcal{M} \setminus \hat{\mathcal{M}} \mid X \cap A = \emptyset\}|}{|\mathcal{M}|}$ for any $A$ which is a partial solution). Lemma~\ref{lem:threshold} demonstrates that such an $(A, \epsilon)$-similar multiset can be constructed by imposing a limit on the sizes of the subsets within $\mathcal{M}$. Furthermore, the lemma guarantees that the average size of the sets in this similar multiset is comparable to the expected runtime of the oracle $\mathcal{Q}_{\ell}$ under that size limit. This relationship serves as the fundamental building block for implementing $\mathcal{Q}_{s}$ via $\mathcal{Q}_{\ell}$.

Recall that Algorithm~\ref{algorithm:linear} starts each time with an avoiding set $A$ and iteratively calls $\mathcal{Q}_s(A)$ until the total size of the sets it draws using $\mathcal{Q}_s$ reaches a threshold, which we denote by $\mathcal{B} \cdot n$. The oracle $\mathcal{Q}_s$ spends $O(\mathcal{B} \cdot n)$ time to construct these subsets. Now, define $\epsilon' = \epsilon/2$ and assume we are given the proper threshold $l$ from Lemma~\ref{lem:threshold} using the parameter $\epsilon'$, and we wish to use it to construct a similar sequence of subsets via $\mathcal{Q}_{\ell}$ using parameter $l$. We claim that if we keep running $\mathcal{Q}_\ell(l)$ until the total query cost reaches $\frac{96 (\lceil \log n \rceil +2)^2}{\epsilon^3} \cdot \mathcal{B} \cdot n$, and only include the non-null sets returned by $\mathcal{Q}_\ell$ that do not contain any element of $A$, we obtain the desired sequence of subsets with constant probability.

To see this, observe that since we use $\epsilon'$ for Lemma~\ref{lem:threshold} and we know that there is at least an $\epsilon$ fraction of the sets in $\mathcal{M}$ that do not intersect with $A$, the guarantees of Lemma~\ref{lem:threshold} imply that at most one half of such subsets have size more than $l$. Thus, at least an $\epsilon'$ fraction of our queries to $\mathcal{Q}_\ell$ result in valid sets. Thus, each time we call $\mathcal{Q}_{\ell}$ with parameter $l$, with probability at least $\epsilon'$,  the outcome results in a valid set. Intuitively, we call a unit of time spent on $\mathcal{Q}_\ell$ \textit{relevant} if it is spent on a set that has size at most $l$ and does not contain $A$, and \textit{not relevant} otherwise. Due to the guarantees of Lemma~\ref{lem:threshold}, the expected ratio of relevant time spent on $\mathcal{Q}_\ell$ over the total time spent on it is at least $\frac{\epsilon'^3}{6 (\lceil \log n \rceil +2)^2}$.

Thus, due to Markov's inequality, after we spent a cost of $\frac{12 (\lceil \log n \rceil +2)^2}{\epsilon'^3} \cdot \mathcal{B} \cdot n = \frac{96 (\lceil \log n \rceil +2)^2}{\epsilon^3} \cdot \mathcal{B} \cdot n$ on $Q_{\ell}$, with probability at least $1/2$ the total size of valid sets is at least $\mathcal{B} \cdot n$ as desired. This signals the fact by replacing $\mathcal{M}$ with an $(A,\epsilon/2)$-similar subset, we can basically implement $Q_s$ via $Q_\ell$ with a multiplicative overhead of $O(\log^2 n/\epsilon^3)$. However, there is one last challenge to be addressed: we are not aware of the correct value of $l$.

Now, suppose we are given a value $l$ and proceed with the aforementioned procedure. However, even if we succeed in generating the desired output, we require a formal guarantee of the validity of our process, particularly in cases where $l$ may not be the correct threshold. While the total sum of valid sets reaching threshold $\mathcal{B} \cdot n$ is an obvious criterion for correctness, it only addresses the second condition of Lemma~\ref{lem:threshold}. If the first condition of Lemma~\ref{lem:threshold} is violated for this $l$, the sampled multiset may not satisfy the necessary similarity properties. Thus, we must verify this condition before declaring success.

To this end, we modify the algorithm by maintaining a counter for the number of sets encountered with size strictly greater than $l$. As long as this counter remains below $200 \log n / \epsilon'$, we utilize the regular oracle $\mathcal{Q}$ instead of $\mathcal{Q}_{\ell}$, while strictly maintaining our cost accounting by only adding $\min\{|X|, l\}$ to our total cost budget in order to keep track of how much $\mathcal{Q}_{\ell}$ would cost on the same samples. Once the counter reaches this threshold, we evaluate the performance on these samples to determine how many of the returned sets do not contain any element of $A$. If this count exceeds $300 \log n$, then with high probability $l$ is not the correct threshold, allowing us to terminate and declare failure. Otherwise, even if $l$ is not the theoretically guaranteed parameter, we know with high probability that the sampled multiset is $(A, \epsilon)$-similar to $\mathcal{M}$. From that point forward, we ignore all sets with size larger than $l$ and exclusively use $\mathcal{Q}_{\ell}$ with parameter $l$, focusing solely on the accumulated sum of valid sets. If our counter never reaches the threshold, we simply consider all subsets (including those with size greater than $l$) and check if their total size is sufficient; if so, we report them, as we are certain we have sampled directly from $\mathcal{M}$.

Our proof for the validity of $l$ introduces three minor trade-offs. First, there is an additional query cost of $O(n \log n / \epsilon)$ for verification of the first condition, which adds an additive $O(\log n / \epsilon)$ term to our existing multiplicative $O(\log^2 n/\epsilon^3)$ overhead which is negligible. Second, it slightly reduces our success probability from $1/2$ to $1/2 - n^{-\Omega(1)}$, which is practically negligible. Third, it introduces a false positive probability of $n^{-\Omega(1)}$, which is well within tolerable limits for our analysis. The pseudocode for this refined procedure is provided in Algorithm~\ref{alg:validation}.

\begin{algorithm}[H]
	\caption{Adaptive Oracle Implementation with Validation}
	\label{alg:validation}
	\begin{algorithmic}[1]
		\STATE \textbf{Input:} Avoiding set $A$, threshold $\mathcal{B} \cdot n$, parameter $l$, error $\epsilon$.
		\STATE \textbf{Initialize:} $S \gets \emptyset$, $\text{total\_cost} \gets 0$, $\text{valid\_sum} \gets 0$, $\text{large\_sets} \gets \emptyset$, $\text{counter} \gets 0$.
		\STATE \textbf{Thresholds:} $T_{cost} \gets \frac{96 (\lceil \log n \rceil + 2)^2}{\epsilon^3} \cdot \mathcal{B} \cdot n$, $T_{large} \gets 400 \log n / \epsilon$.
		
		\WHILE{$\text{total\_cost} < T_{cost}$ \textbf{and} $\text{valid\_sum} < \mathcal{B} \cdot n$}
		\IF{$\text{counter} < T_{large}$}
		\STATE Draw $X \sim \mathcal{M}$ using regular oracle $\mathcal{Q}$.
		\STATE $\text{total\_cost} \gets \text{total\_cost} + \min\{|X|, l\}$.
		\IF{$X \cap A = \emptyset$}
		\STATE $S \gets S \cup \{X\}$, $\text{valid\_sum} \gets \text{valid\_sum} + |X|$.
		\ENDIF
		\IF{$|X| > l$}
		\STATE Add $X$ to $\text{large\_sets}$ and $\text{counter} \gets \text{counter} + 1$.
		\IF{$\text{counter} = T_{large}$}
		\IF{$|\{X \in \text{large\_sets} \mid X \cap A = \emptyset\}| > 300 \log n$}
		\RETURN \textbf{Failure} (Threshold $l$ is invalid).
		\ELSE
		\FOR{$X \in \{\text{large\_sets} \mid X \cap A = \emptyset\}$}
		\STATE $S \gets S \setminus \{X\}$, $\text{valid\_sum} \gets \text{valid\_sum} - |X|$.
		\ENDFOR
		\ENDIF
		\ENDIF
		\ENDIF
		\ELSE
		\STATE Draw $X$ using $\mathcal{Q}_{\ell}$ with parameter $l$.
		\STATE $\text{total\_cost} \gets \text{total\_cost} + \text{cost}(\mathcal{Q}_{\ell})$.
		\IF{$X \neq \text{null}$ \textbf{and} $X \cap A = \emptyset$}
		\STATE $S \gets S \cup \{X\}$, $\text{valid\_sum} \gets \text{valid\_sum} + |X|$.
		\ENDIF
		\ENDIF
		\ENDWHILE
		
		\IF{$\text{valid\_sum} \geq \mathcal{B} \cdot n$}
		\RETURN \textbf{Success}, $S$.
		\ELSE
		\RETURN \textbf{Failure}.
		\ENDIF
	\end{algorithmic}\label{alg:redu}
\end{algorithm}

We are now ready to state the main result of this section.
\begin{lemma}\label{lemma:gg}
	Let $A \subseteq \mathcal{N}$ be a subset of elements such that more than an $\epsilon$ fraction of sets in $\mathcal{M}$ do not contain any element from $A$. For a given budget $\mathcal{B} \cdot n$, one can generate a sequence of subsets $S_1, S_2, \ldots$ such that 
	\begin{itemize}
		\item $\sum |S_i| \geq \mathcal{B} \cdot n$.
		\item $S_i$'s are drawn uniformly and independently from a multiset which is $(A,\epsilon)$-similar to $\mathcal{M}$.
	\end{itemize}
	The algorithm runs in time $O(\frac{\log^4 n}{\epsilon^3}\cdot \mathcal{B} \cdot n)$  and succeeds with probability at least $1-n^{-9}$.
\end{lemma}
\begin{proof}
	We have established that Algorithm~\ref{alg:redu} serves as a validated implementation of $\mathcal{Q}_s$ via $\mathcal{Q}_{\ell}$. To complete the proof, we formally address the success probability, the false positive rate, and the total complexity.
	
	By Lemma~\ref{lem:threshold}, there exists at least one value $l^* \in \{0, 1, 3, \ldots, 2^{\lceil \log_2 n \rceil + 1}-1\}$ that satisfies both conditions for $\epsilon' = \epsilon/2$. Since there are $\lceil \log_2 n \rceil + 2$ possible values in this set, a uniform random guess for $l$ selects the correct threshold with probability $1/(\lceil \log_2 n \rceil + 2)$.
	
	When the correct $l^*$ is chosen, we must account for two failure modes: first, that we fail to accumulate enough valid sets within the cost budget, and second, that our verification step erroneously rejects the correct $l^*$. 
	
	\begin{itemize}
		\item Budget Sufficiency: Based on the relevant time ratio established by Lemma~\ref{lem:threshold}, the expected cost to reach $\mathcal{B} \cdot n$ is half of our budget, which is $\frac{48 (\lceil \log n \rceil + 2)^2}{\epsilon^3} \cdot \mathcal{B} \cdot n$. By Markov's inequality, the probability that the actual cost exceeds the total budget of $\frac{96 (\lceil \log n \rceil + 2)^2}{\epsilon^3} \cdot \mathcal{B} \cdot n$ before we accumulate the required valid size is at most $1/2$.
		
		\item Verification Reliability: If $l^*$ is correct, the first condition of Lemma~\ref{lem:threshold} guarantees that the probability a set $X \in \mathcal{M}$ with $|X| > l^*$ does not contain an element of $A$ is at most $\epsilon/2$. Let $\Xi$ be the number of sets not intersecting $A$ among the $\kappa = 400 \log n / \epsilon$ samples of size $> l^*$. Then $\Xi$ follows a Binomial distribution $B(\kappa, \rho)$ where $\rho \leq \epsilon/2$. The expected value is $\mu = \kappa \rho \leq 200 \log n$. To erroneously reject $l^*$, the number of non-intersections must exceed $300 \log n$. Using the Chernoff bound $\mathbb{P}(\Xi \geq (1+\delta)\mu) \leq e^{-\frac{\delta^2 \mu}{2+\delta}}$ with $\delta = 0.5$ and $\mu = 200 \log n$, we have:
		\[ \mathbb{P}(\Xi \geq 300 \log n) \leq \exp\left(-\frac{0.25 \cdot 200 \log n}{2.5}\right) = \exp(-20 \log n) \le \exp(-20 \ln n) = n^{-20}. \]
	\end{itemize}
	
	Thus, a single invocation with the correct $l^*$ succeeds with probability at least $1/2 - n^{-20} > 1/3$. Including the probability of guessing $l$, the success probability is $\Omega(1/\log n)$. By running $O(\log^2 n)$ trials, we succeed with probability at least $1 - (1 - \frac{1}{2(\lceil \log n \rceil + 2)})^{O(\log^2 n)} \geq 1-n^{-10}$.
	
	A false positive occurs if the algorithm declares success for an invalid $l$. If $l$ violates the first condition of Lemma~\ref{lem:threshold} in that more than an $\epsilon$ fraction of the sets in $\mathcal{M}$ that have size greater than $l$ do not intersect with $A$, the expected number of sets not intersecting $A$ in the validation samples is $\mu \geq 400 \log n$. The probability that we see $\Xi \leq 300 \log n$ sets not intersecting $A$ is bounded by the Chernoff bound $\mathbb{P}(\Xi \leq (1-\delta)\mu) \leq e^{-\frac{\delta^2 \mu}{2}}$. With $\delta = 100/400 = 0.25$ and $\mu = 400 \log n$, we get:
	\[ \mathbb{P}(\Xi \leq 300 \log n) \leq \exp\left(-\frac{0.0625 \cdot 400 \log n}{2}\right) = \exp(-12.5 \log n) \le \exp(-12.5 \ln n) = n^{-12.5}. \]
	This ensures the overall false positive probability across all $O(\log^2 n)$ trials is at most $n^{-10}$, which is well within our $n^{-9}$ budget.
	
	Each invocation of Algorithm~\ref{alg:redu} is capped by a cost of $\frac{96 (\lceil \log n \rceil + 2)^2}{\epsilon^3} \cdot \mathcal{B} \cdot n$. Since we perform $O(\log^2 n)$ trials to cover the candidate values of $l$, the total runtime is:
	\[ O\left( \log^2 n \cdot \frac{(\lceil \log n \rceil + 2)^2}{\epsilon^3} \cdot \mathcal{B} n\right) = O\left(\frac{\log^4 n}{\epsilon^3} \cdot \mathcal{B} n\right). \]
	This completes the proof.
\end{proof}

\begin{algorithm}[H]\label{alg:rere}
	\caption{Implementation of $\mathcal{Q}_s$ via $\mathcal{Q}_{\ell}$}
	\begin{algorithmic}[1]
		\STATE \textbf{Initialize:} $L = \{0, 1, 3, \ldots, 2^{\lceil \log_2 n \rceil + 1}-1\}$, $\text{Trials} = c \cdot \log^2 n$.
		\FOR{$i = 1$ \TO $\text{Trials}$}
		\STATE Pick $l$ uniformly at random from $L$.
		\STATE {Result},$S$ $\gets$ \text{AdaptiveOracleWithValidation}$(A, \mathcal{B} \cdot n, l, \epsilon)$.
		\IF{{Result} is \textbf{Success}}
		\RETURN $S$.
		\ENDIF
		\ENDFOR
		\STATE \text{report $A$ as the solution and terminate the algorithm}.
	\end{algorithmic}
\end{algorithm}

The algorithm explained in the proof of Lemma~\ref{lemma:gg} is shown in Algorithm~\ref{alg:rere}. Lemma~\ref{lemma:gg} implies that in every phase of Algorithm~\ref{algorithm:linear}, the additive difference between the total perceived contribution of the elements and their actual contribution is bounded by an $\epsilon$ multiplicative factor of the optimal solution. Thus, we can extend the proof of Theorem~\ref{theorem:linear} to work for this setting as well.

\begin{theorem}\label{theorem:final}
	If query oracle $\mathcal{Q}_{\ell}$ is available and is used via Algorithm~\ref{alg:rere} to implement $\mathcal{Q}_s$, then Algorithm~\ref{algorithm:linear}  runs in time $O(\frac{n \log^{14} n  \log (1/\epsilon)}{\epsilon^5})$, has query cost $O(\frac{n \log^{13} n  \log (1/\epsilon)}{\epsilon^5})$, and achieves a $\left( 1 - \frac{1}{e} - (4w+7)\epsilon \right)$-approximate solution for the maximum $k$-cover problem with probability at least $1-n^{-7}$ where $w$ is the number of phases of Algorithm~\ref{algorithm:linear}, provided that $\epsilon \le 1/(50w)$.
\end{theorem}
\begin{proof}
	 Since our algorithm may face $O(\log n)$ different distributions, in order to apply Lemma~\ref{lemma:kuchik} we multiply the number of allowed phases by a factor of $O(\log n)$. \claudeedit{Indeed, the samples of a phase are drawn uniformly from the sets of size at most $l$ that do not intersect $A$, i.e., from the residual problem of the coverage function defined by the sets of size at most $l$, and this function depends only on $l$. Hence, we can apply the argument of Lemma~\ref{lemma:kuchik} separately to each of the $O(\log n)$ coverage functions defined by the sets of size at most $l$, since the phases that use a different threshold only add elements to the solution and thus never increase the values $\Delta_x$ defined for $l$.} This in addition to the $O(\log^4 n / \epsilon ^ 3)$ overhead of Lemma~\ref{lemma:gg} results in query complexity $O(\frac{n \log^{13} n  \log (1/\epsilon)}{\epsilon^5})$ and runtime $O(\frac{n \log^{14} n  \log (1/\epsilon)}{\epsilon^5})$. See the proof of Theorem~\ref{theorem:linear} for more details. 
	 
	 In the rest of the proof, we show that the approximation ratio of the algorithm is at least $\left( 1 - \frac{1}{e} - (4w+7)\epsilon \right)$. The proof is almost identical to that of Theorem~\ref{theorem:linear} except that here we have an extra error corresponding to oracle $Q_{\ell}$. We assume for simplicity here that the algorithm has no failure and discuss the failure probability at the end of the proof. Also, we assume for simplicity here that Algorithm~\ref{algorithm:linear} iteratively puts $k$ elements into the solution and the early termination does not happen. We discuss at the end of the proof that the early termination does not deteriorate the quality of our solution.
	
	To this end, let $d_1, d_2, \ldots, d_k$ be the marginal contributions of the elements that Algorithm~\ref{algorithm:linear} iteratively puts into the solution to the function $\sigma(\mathcal{C})$. Also, we define $d'_i$ as the same contribution for the sampled problem of the corresponding round. That is, let $F'$ be the sampled $k$-cover problem of the corresponding round. Instead of function $F$, we use $F'$ to determine the value of $d'_i$. We also define $u$ as the highest coverage of any subset of size $k$ which is the coverage ratio of the optimal solution.

	Now, we wish to apply Lemma~\ref{lem:lower_bound_sum} to give a lower bound on the quality of our algorithm. Fix some $1 \leq r \leq k$ and assume that the $r$'th element of the solution is added in round $j$. Also, assume that $q$ elements are added to the solution prior to round $j$. Define $v$ as the highest utility a set of size $k$ can obtain in the beginning of round $j$ in the residual problem. It follows that $v \geq u - \sum_{i=1}^q  d_i$. Also, define $v'$ as the coverage ratio of the same subset in the sampled problem. Based on Corollary~\ref{lemma:cor2} and Lemma~\ref{lemma:gg}, we have $v' \geq v (1-\epsilon) - 2 \epsilon u \geq (u - \sum_{i=1}^q  d_i)(1-\epsilon) - 2\epsilon u $. Moreover, it follows from Corollary~\ref{lemma:cor} and Lemma~\ref{lemma:gg} that $\sum_{i=1}^q  d_i \leq (1+\epsilon) \sum_{i=1}^q  d'_i + 2w\epsilon u$ and therefore $v' \geq (u - (1+\epsilon)\sum_{i=1}^q  d'_i)(1-\epsilon) - (2w+2)\epsilon u$. This implies that at the time of adding the $r$'th element to the solution, there is a subset of size $k$ whose empirical contribution with respect to the sampled problem is at least 
	\begin{align*}
		v' - \sum_{i=q+1}^{r-1} d'_i - (2w+2)\epsilon u&\geq \left( u - (1+\epsilon)\sum_{i=1}^q d'_i \right)(1-\epsilon) - \sum_{i=q+1}^{r-1} d'_i - (2w+2)\epsilon u\\
		&= u(1-\epsilon) - (1-\epsilon^2)\sum_{i=1}^q d'_i - \sum_{i=q+1}^{r-1} d'_i - (2w+2)\epsilon u\\
		&\geq u(1-\epsilon) - \sum_{i=1}^q d'_i - \sum_{i=q+1}^{r-1} d'_i - (2w+2)\epsilon u\\
		&= u(1-\epsilon) - \sum_{i=1}^{r-1} d'_i - (2w+2)\epsilon u\\
		&\geq u(1-\epsilon) - (1+2\epsilon)(1-\epsilon) \sum_{i=1}^{r-1} d'_i - (2w+2)\epsilon u\\ 
		&= \left( u - (1+2\epsilon)\sum_{i=1}^{r-1} d'_i \right)(1-\epsilon) - (2w+2)\epsilon u\\
		&\geq \left( u - (1+(4w+6)\epsilon)\sum_{i=1}^{r-1} d'_i \right)(1-(4w+5)\epsilon).
	\end{align*}

Since each time we add the element with the highest marginal contribution, we have $d'_i \geq \left( u - (1+(4w+6)\epsilon)\sum_{i=1}^{r-1} d'_i \right)\frac{(1-(4w+5)\epsilon)}{k}$. Applying Lemma~\ref{lem:lower_bound_sum} to $d'$ with $\beta = 1+(4w+6)\epsilon$ and $\alpha = 1-(4w+5)\epsilon$ implies $\sum_{i=1}^k d'_i \ge u \left( 1 - \frac{1}{e} \right)  (1-(4w+6)\epsilon)$. This in addition to Corollary~\ref{lemma:cor} implies $\sum_{i=1}^k d_i \ge u \left( 1 - \frac{1}{e} \right)  (1-(4w+6)\epsilon)(1-\epsilon) \geq  u \left( 1 - \frac{1}{e} - (4w+7)\epsilon \right)$ which is desired.
	
	In the case of early termination by crossing the phase limit threshold, we know that the contribution of each neglected element is bounded by $\epsilon / n$ and therefore our solution has a quality of at least $(1-\epsilon)$ which is far more desirable than  the $\left( 1 - \frac{1}{e} - (4w+7)\epsilon \right)$ approximation factor. This is also the case for early termination of Algorithm~\ref{alg:rere} which only happens when our current solution covers at least $1-\epsilon$ fraction of the sets.
	
	Finally, since each failure probability is bounded by $n^{-9}$ the number of such possible scenarios is bounded by $n^2$, the total failure probability of the algorithm is bounded by $n^{-7}$.
\end{proof}

Finally, it follows from Theorem~\ref{theorem:final} that by scaling $\epsilon$, one can approximate the influence maximization problem within an approximation factor of $1-1/e-\epsilon$ in time $\tilde O_{\epsilon}(n+m)$. Since $w = O(\log^6 n \log(1/\epsilon))$, it suffices to run the algorithm with error parameter $\epsilon/\Theta(w)$, which increases the running time only by a polylogarithmic factor.

\begin{corollary}[of theorem~\ref{theorem:final}]\label{cor:final}
		The influence maximization problem can be approximated within a factor of $1-1/e-\epsilon$ in time $\tilde O_{\epsilon}(n+m)$. The algorithm has a failure probability of $n^{-7}$.
\end{corollary}

\subsection{Proof of Lemma~\ref{lem:threshold}}\label{sec:prooflemmal}

Before we prove Lemma~\ref{lem:threshold}, we state Lemma~\ref{lem:average_ratio} as an auxiliary observation.

\begin{lemma}
	\label{lem:average_ratio}
	Let $X$ be a multiset of numbers in $\{1, \ldots, n\}$ and let $Y \subseteq X$. Suppose there exists a constant $0 < \alpha \leq 1$ such that for every dyadic scale $i \in \{1, 2, 4, \ldots, 2^{\lceil \log_2 n \rceil}\}$, the following tail density condition holds:
	$$ |\{y \in Y \mid y \geq i\}| \geq \alpha \cdot |\{x \in X \mid x \geq i\}| $$
	Then the average of the elements in $Y$ is at least an $\alpha/2$ fraction of the average of the elements in $X$.
\end{lemma}

\begin{proof}
Let $T = \{2^0, 2^1, \dots, 2^{\lceil \log_2 n \rceil}\}$. We define a transformed multiset $Y'$ such that for each $y \in Y$, the corresponding $y' \in Y'$ is the smallest power of two larger than $y$, minus one. That is, if $y \in [2^k, 2^{k+1}-1]$, then $y' = 2^{k+1}-1$. Also, we define $\Sigma_X$, $\Sigma_Y$, and  $\Sigma_{Y'}$ to denote the sum of elements in $X$, $Y$, and $Y'$ respectively.
	We observe that for $Y'$, on each $t \in T$, the tail count $|\{y' \in Y' \mid y' \geq \gamma\}|$ is the same for every $t \leq \gamma < 2t$. More precisely, for any $t \in T$, we have $|\{y' \in Y' \mid y' \geq \gamma\}| = |\{y' \in Y' \mid y' \geq t\}|$ for every $t \leq \gamma < 2t$. Thus we can formulate $\sum_{Y'}$ using the identity $\sum_{z \in Z} z = \sum_{t \geq 1} |\{z \in Z \mid z \geq t\}|$ as follows:
	$$ \Sigma_{Y'} =  \sum_{t=1}^{2^{\lceil \log_2 n \rceil+1}-1} |\{y' \in Y' \mid y' \geq t\}|  = \sum_{\gamma \in T} \sum_{t=\gamma}^{2\gamma-1} |\{y' \in Y' \mid y' \geq t\}| = \sum_{\gamma \in T} \gamma \cdot |\{y \in Y \mid y \geq \gamma\}|.$$
	Now we provide an upper bound for $\Sigma_X$ using the same dyadic decomposition:
	$$ \Sigma_X = \sum_{\gamma \in T} \sum_{t=\gamma}^{2\gamma-1} |\{x \in X \mid x \geq t\}| \leq \sum_{\gamma \in T} \gamma \cdot |\{x \in X \mid x \geq \gamma\}|.$$
	Applying the tail density assumption $|\{y \in Y \mid y \geq \gamma\}| \geq \alpha |\{x \in X \mid x \geq \gamma\}|$ to every term in the sum:
	$$ \Sigma_{Y'} = \sum_{\gamma \in T} \gamma \cdot |\{y \in Y \mid y \geq \gamma\}| \geq \alpha \sum_{\gamma \in T} \gamma \cdot |\{x \in X \mid x \geq \gamma\}| \geq \alpha \Sigma_X.$$
	Since $y' \le 2y$ for all $y \ge 1$ (specifically, $y' = 2^{k+1}-1 < 2 \cdot 2^k \le 2y$), the sum $\Sigma_{Y'}$ is strictly less than $2\Sigma_Y$. This implies that 
	$ 2\Sigma_Y > \Sigma_{Y'} \geq \alpha \Sigma_X$ and thus $\Sigma_Y \geq \frac{\alpha}{2} \Sigma_X $ holds.
	Let $E_X = \Sigma_X / |X|$ and $E_Y = \Sigma_Y / |Y|$ be the averages. Since $Y \subseteq X$, it follows that $|Y| \leq |X|$. Thus:
	$$ E_Y = \frac{\Sigma_Y}{|Y|} \geq \frac{\frac{\alpha}{2} \Sigma_X}{|Y|} \geq \frac{\alpha}{2} \frac{\Sigma_X}{|X|} = \frac{\alpha}{2} E_X.$$
	This completes the proof.
\end{proof}

Now we are ready to prove Lemma~\ref{lem:threshold}.

\begin{proof}[of Lemma~\ref{lem:threshold}]
	Let $T = \{2^0, 2^1, \ldots, 2^{\Delta-1}\}$ be the set of dyadic scales. For $\gamma \in T$, let $\mathcal{M}_{\ge \gamma} = \{X \in \mathcal{M} \mid |X| \ge \gamma\}$ and similarly, $\mathcal{M}^*_{\ge \gamma} = \{X \in \mathcal{M}^* \mid |X| \ge \gamma\}$. Also let $S^*_\gamma = \{X \in \mathcal{M}^* \mid \gamma \le |X| < 2\gamma\}$. Define the density parameter $\epsilon' = \frac{\epsilon}{\Delta}$.
	

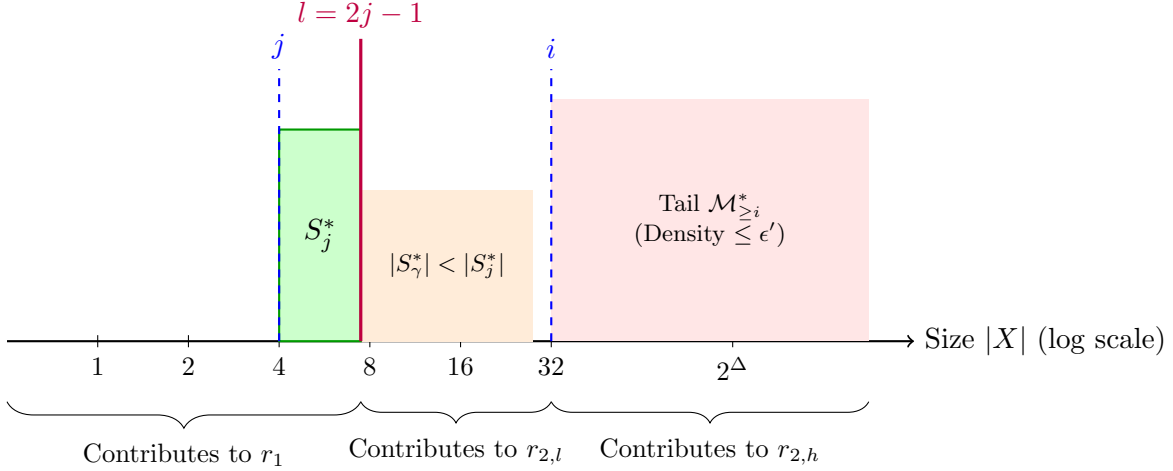
\begin{figure}[htbp]
	\centering
	\begin{tikzpicture}[xscale=1.2, yscale=0.8]
		\draw[thick, ->] (0,0) -- (10,0) node[right] {Size $|X|$ (log scale)};
		
		\draw (1, 0.1) -- (1, -0.1) node[below, font=\small] {1};
		\draw (2, 0.1) -- (2, -0.1) node[below, font=\small] {2};
		\draw (3, 0.1) -- (3, -0.1) node[below, font=\small] {4};
		\draw (4, 0.1) -- (4, -0.1) node[below, font=\small] {8};
		\draw (5, 0.1) -- (5, -0.1) node[below, font=\small] {16};
		\draw (6, 0.1) -- (6, -0.1) node[below, font=\small] {32};
		\draw (8, 0.1) -- (8, -0.1) node[below, font=\small] {$2^\Delta$};
		
		\fill[red!10] (6,0) rectangle (9.5, 4);
		\node[black, font=\footnotesize, align=center] at (7.75, 2) {Tail $\mathcal{M}^*_{\ge i}$\\(Density $\le \epsilon'$)};
		
		\fill[orange!15] (3.9,0) rectangle (5.8, 2.5);
		\node[black, font=\footnotesize] at (4.85, 1.25) {$|S^*_\gamma| < |S^*_j|$};
		
		\fill[green!20] (3,0) rectangle (3.9, 3.5);
		\draw[thick, green!60!black] (3,0) rectangle (3.9, 3.5);
		\node[black, font=\small, font=\bfseries] at (3.45, 1.75) {$S^*_j$};
		
		\draw[dashed, blue, thick] (6, 0) -- (6, 4.5) node[above] {$i$};
		\draw[dashed, blue, thick] (3, 0) -- (3, 4.5) node[above] {$j$};
		\draw[very thick, purple] (3.9, 0) -- (3.9, 5) node[above] {$l = 2j-1$};
		
		\draw[decorate, decoration={brace, amplitude=10pt, mirror}] (0,-1) -- (3.9,-1) 
		node[midway, below=12pt, font=\small] {Contributes to $r_1$};
		
		\draw[decorate, decoration={brace, amplitude=8pt, mirror}] (3.9,-1) -- (6,-1) 
		node[midway, below=10pt, font=\small] {Contributes to $r_{2,l}$};
		
		\draw[decorate, decoration={brace, amplitude=8pt, mirror}] (6,-1) -- (9.5,-1) 
		node[midway, below=10pt, font=\small] {Contributes to $r_{2,h}$};
		
	\end{tikzpicture}
	\caption{Decomposition of the multiset mass into components $r_1$, $r_{2,l}$, and $r_{2,h}$ relative to the chosen threshold $l$ and the tail threshold $i$.}
	\label{fig:threshold_decomposition}
\end{figure}

	Let $i$ be the smallest integer in $T$ such that $|\mathcal{M}^*_{\ge i}| \leq \epsilon' |\mathcal{M}_{\ge i}|$. If no such $i$ exists, we set $l = 2^{\Delta-1}-1$ and since both $\mathcal{M}_{\geq l+1}$ and $\mathcal{M}^*_{\geq l+1}$ are empty the first condition would hold. Also, the second condition holds since Lemma~\ref{lem:average_ratio} implies that the average size of sets in $\mathcal{M}^*$ is at least an $\epsilon'/2 \geq \frac{\epsilon^2}{6 \Delta^2}$ fraction of the average size of the sets in $\mathcal{M}$. For all $i' \in T$ such that $i' < i$, the minimality of $i$ implies:
	$ \frac{|\mathcal{M}^*_{\ge i'}|}{|\mathcal{M}_{\ge i'}|} > \epsilon' $. This ensures that for all scales strictly smaller than $i$, the ratio of the number of sets in $\mathcal{M}^*$ to those in $\mathcal{M}$ is bounded from below by $\epsilon'$. This also implies that for all $i' \in T$ such that $i' < i$ we have:
	\begin{equation}\label{eq:mohem1}
		\frac{|\mathcal{M}^*_{\ge i'} \setminus \mathcal{M}^*_{\ge i}|}{|\mathcal{M}_{\ge i'} \setminus \mathcal{M}_{\ge i}|} > \epsilon'
	\end{equation}
	
	Because  $|\mathcal{M}^*| > \epsilon |\mathcal{M}|$ and $\epsilon' = \frac{\epsilon}{\Delta}$ there exists a $\gamma \in T$ such that $\gamma < i$ and $|S^*_\gamma| > \epsilon'|\mathcal{M}_{\geq i}|$. We select $j$ to be the largest such $\gamma$. Based on maximality of $j$ and the fact that the number of elements between (but not including) $j$ and $i$ in $T$ is at most $\Delta$ we have:
	\begin{equation}\label{eq:mohem2}
		|S^*_j| \ge \frac{1}{\Delta} \sum_{\gamma \in T, j < \gamma < i} |S^*_\gamma|.
	\end{equation}
	Now, we define the threshold for the lemma as $l = 2j - 1$ and prove the conditions of the lemma below.
	
	By our choice of $l$, we have $\{X \in \mathcal{M}^* \mid |X| > l\} = \mathcal{M}^*_{\ge l+1}$. We partition $\mathcal{M}^*_{\ge l+1}$ into the tail at $i$ and the intermediate dyadic intervals:
	\begin{equation}\label{eq:notimportant}
		|\mathcal{M}^*_{\ge l+1}| = |\mathcal{M}^*_{\ge i}| + \sum_{\gamma \in T, j < \gamma < i} |S^*_\gamma| 
	\end{equation}
	Each of the at most $\Delta$ terms that are summed up in the right hand side of Equation~\eqref{eq:notimportant} is bounded by $\epsilon' |\mathcal{M}_{\geq i}|$ and thus $$ |\mathcal{M}^*_{\ge l+1}| \leq \epsilon' \Delta |\mathcal{M}_{\geq i}| = \epsilon |\mathcal{M}_{\geq i}| \leq \epsilon |\mathcal{M}_{\ge l+1}|$$ which is the first condition of the lemma.

If we categorize the sets in $\mathcal{M}$ by whether or not their size exceeds $l$, we can imply that the left hand side of the second condition of the lemma is equal to 
$$\mathbb{E}_{X \sim \mathcal{M}}[\min\{|X|, l\}] = \mathbb{E}_{X \sim \mathcal{M}}[|X| \cdot \mathbb{I}_{|X| \leq l}] + l \cdot \mathbb{P}_{X \sim \mathcal{M}}[|X| > l]$$
Let us break the right hand side into two terms $r_1 = \mathbb{E}_{X \sim \mathcal{M}}[|X| \cdot \mathbb{I}_{|X| \leq l}] $ and $r_2 = l \cdot \mathbb{P}_{X \sim \mathcal{M}}[|X| > l]$. Define $\mathcal{M}_s = \mathcal{M} \setminus \mathcal{M}_{\geq l+1}$ and similarly $\mathcal{M}^*_s = \mathcal{M}^* \setminus \mathcal{M}^*_{\geq l+1}$. Inequalities~\eqref{eq:mohem1} and ~\eqref{eq:mohem2} together imply that conditions of Lemma~\ref{lem:average_ratio} hold for $\mathcal{M}_s$ and $\mathcal{M}^*_s$ with parameter $\alpha = \epsilon'/\Delta$ and thus $r_1$ is bounded by $2\Delta/\epsilon' = 2\Delta^2/\epsilon$ times the average size of sets in $\mathcal{M}^*_s$.
For $r_2$, we conceptually divide $\mathbb{P}_{X \sim \mathcal{M}}[|X| > l]$ into two terms $r_{2,h} =  l \cdot \mathbb{P}_{X \sim \mathcal{M}}[|X| \geq i]$ and $r_{2,l} =  l \cdot \mathbb{P}_{X \sim \mathcal{M}}[l < |X|<i]$.

For $r_{2,l}$ we leverage Inequalities~\eqref{eq:mohem1} and ~\eqref{eq:mohem2} to imply that $\frac{|S^*_j|}{|\mathcal{M}|} \geq \frac{\epsilon'}{\Delta} \mathbb{P}_{X \sim \mathcal{M}}[l < |X|<i]$ and thus $r_{2,l} \leq \frac{2\Delta^2}{\epsilon} \mathbb{E}_{X \sim \mathcal{M}^* \text{ such that }|X| \leq l}[|X|]$. Also, $j$ is chosen such that $|S^*_j| > \epsilon' |\mathcal{M}_{\geq i}|$ and thus $r_{2,h} \leq \frac{2\Delta}{\epsilon} \mathbb{E}_{X \sim \mathcal{M}^* \text{ such that }|X| \leq l}[|X|]$. Therefore $r_1 + r_{2,l} + r_{2,h} \leq \frac{6 \Delta^2}{\epsilon^2} \cdot \mathbb{E}_{X \sim \mathcal{M}^* \text{ such that }|X| \leq l}[|X|]$ which is desired.
\end{proof}